\RequirePackage{amsmath}

\newif\ifsubmit     
\newif\ifllncs      
\newif\ifexabs      
\newif\ifblind

\submittrue

\ifllncs
  \documentclass[runningheads,a4paper]{llncs}

\else
  \documentclass[letterpaper,11pt,pdfa]{article}
  \usepackage[in]{fullpage}
\fi

\usepackage{iftex}
\ifPDFTeX
  \usepackage[utf8]{inputenc}
  \usepackage[noTeX]{mmap}
  \usepackage[T1]{fontenc}
\fi
\ifLuaTeX
  \usepackage{luatex85}
  \usepackage[noTeX]{mmap}
\fi

\usepackage{mdframed}
\usepackage{amsmath}
\usepackage{amsfonts}
\usepackage{amssymb}
\usepackage{amsthm}
\usepackage{color}
\usepackage{hhline}

\usepackage{appendix}
\usepackage{algorithm}
\usepackage{algpseudocode}
\usepackage{braket}
\usepackage{comment}
\usepackage{url}
\usepackage{bbm}
\usepackage{multicol}
\usepackage{mdframed}
\usepackage{multirow}

\usepackage{breakcites}

\usepackage[bookmarks]{hyperref}
\usepackage[nameinlink]{cleveref}
\hypersetup{
  colorlinks,
  allcolors=blue
}

\usepackage{tikz}

\ifllncs

  \spnewtheorem{claim}{Claim}{\bfseries}{\rmfamily}
  \crefname{claim}{claim}{claims}
  \Crefname{claim}{Claim}{Claims}
\else
  \newtheorem{theorem}{Theorem}[section]
  \newtheorem{definition}[theorem]{Definition}
  \newtheorem{remark}[theorem]{Remark}
  \newtheorem{lemma}[theorem]{Lemma}
  \newtheorem{corollary}[theorem]{Corollary}

  \newtheorem{claim}[theorem]{Claim}
  \newtheorem*{remark*}{Remark}

  \newtheorem*{theorem*}{Theorem}
  \newtheorem*{lemma*}{Lemma}

\usepackage[style=alphabetic,minalphanames=3,maxalphanames=4,maxnames=99,backref=true]{biblatex}

  \DeclareFieldFormat{eprint:iacr}{Cryptology ePrint Archive: \href{https://ia.cr/#1}{\texttt{#1}}}
  \DeclareFieldFormat{eprint:iacrarchive}{Cryptology ePrint Archive: \href{https://eprint.iacr.org/archive/#1}{\texttt{#1}}}
  \AtEveryBibitem{
    \clearlist{address}
    \clearfield{date}
    \clearfield{isbn}
    \clearfield{issn}
    \clearlist{location}
    \clearfield{month}
    \clearfield{series}
 
    \ifentrytype{book}{}{
      \clearlist{publisher}
      \clearname{editor}
    }
  }
\fi

\usepackage{appendix}
\usepackage{algorithm}
\usepackage{algpseudocode}
\usepackage{braket}
\usepackage{comment}
\usepackage{url}
\usepackage{multicol}
\usepackage{tikz}
\usetikzlibrary{arrows,automata,positioning}
\usepackage{caption}
\usepackage[caption=false]{subfig}
\usepackage[font=small,labelfont=bf]{caption}
\usepackage{comment}
\usepackage{pifont}
\usetikzlibrary{patterns}

\usepackage{enumitem}
\setlist[description]{noitemsep}
\setlist[enumerate]{noitemsep}
\setlist[itemize]{noitemsep}

\usepackage{soul, xcolor, xparse}
\makeatletter
  \ExplSyntaxOn
    \cs_new:Npn \white_text:n #1
    {
      \fp_set:Nn \l_tmpa_fp {#1 * .01}
      \llap{\textcolor{white}{\the\SOUL@syllable}\hspace{\fp_to_decimal:N \l_tmpa_fp em}}
      \llap{\textcolor{white}{\the\SOUL@syllable}\hspace{-\fp_to_decimal:N \l_tmpa_fp em}}
    }
    \NewDocumentCommand{\whiten}{ m }
    {
      \int_step_function:nnnN {1}{1}{#1} \white_text:n
    }
  \ExplSyntaxOff
  
  \NewDocumentCommand{ \varul }{ D<>{5} O{0.2ex} O{0.1ex} +m } {%
    \begingroup
    \setul{#2}{#3}%
    \def\SOUL@uleverysyllable{%
      \setbox0=\hbox{\the\SOUL@syllable}%
      \ifdim\dp0>\z@
      \SOUL@ulunderline{\phantom{\the\SOUL@syllable}}%
      \whiten{#1}%
      \llap{%
        \the\SOUL@syllable
        \SOUL@setkern\SOUL@charkern
      }%
      \else
      \SOUL@ulunderline{%
        \the\SOUL@syllable
        \SOUL@setkern\SOUL@charkern
      }%
      \fi}%
    \ul{#4}%
    \endgroup
  }
\makeatother

\usepackage{mleftright}

\newcommand{\norm}[1]{\left\lVert#1\right\rVert}

\newcommand{\As}{\mathcal{A}}

\newcommand{\Sim}{{\sf Sim}}

\newcommand{\cA}{\mathcal{A}}

\newcommand{\cB}{\mathcal{B}}

\newcommand{\cC}{\mathcal{C}}

\mdfdefinestyle{figstyle}{ 
  linecolor=black!7, 
  backgroundcolor=black!7, 
  innertopmargin=10pt, 
  innerleftmargin=25pt, 
  innerrightmargin=25pt, 
  innerbottommargin=10pt 
}

\renewcommand{\kappa}{\ell}

\newcommand{\ketbra}[2]{\ket{#1}\!\bra{#2}}

\newcommand\numeq[1]%
  {\stackrel{\scriptscriptstyle(\mkern-1.5mu#1\mkern-1.5mu)}{=}}

\newcommand{\removed}[1]{}

\newcommand{\xo}[1]{{x}_{{\sf o}, #1}}
\newcommand{\yo}[1]{{y}_{{\sf o}, #1}}
\newcommand{\xout}{\vec{x}_{\sf o}}
\newcommand{\yout}{\vec{y}_{\sf o}}

\newcommand{\cH}{\mathcal{H}}

\newcommand{\cX}{\mathcal{X}}

\newcommand{\cY}{\mathcal{Y}}

\newcommand{\cZ}{\mathcal{Z}}

\newcommand{\multigame}{\allowbreak multi-output $k$-search game } %

\newcommand{\gamemath}{\mathcal{G}}

\newcommand{\sym}{\mathcal{S}}

\newlength{\saveparindent}
\newlength{\saveparskip}
\newcounter{ctr}

\newcounter{ectr}

\newcommand{\ignore}[1]{}  

\usepackage{titling}

\title{
NISQ Security and Complexity via \\ Simple Classical Reasoning
}

\author{Alexandru Cojocaru \thanks{University of Edinburgh}
\and Juan Garay \thanks{Texas A\&M University}
\and Qipeng Liu \thanks{UC San Diego}
\and Fang Song \thanks{Portland State University}
}

\date{}

\begin{document}

\maketitle

\begin{abstract}

We give novel lifting theorems
for security games in the quantum random oracle model (QROM) in Noisy Intermediate-Scale Quantum (NISQ) settings such as the hybrid query model, the noisy oracle and the bounded-depth models.
We provide, for the first time, a \emph{hybrid lifting theorem} for hybrid algorithms that can perform both quantum and classical queries, as well as a lifting theorem for quantum algorithms with access to noisy oracles or bounded quantum depth. 

At the core of our results lies a novel measure-and-reprogram framework, called {\em hybrid coherent measure-and-reprogramming}, tailored specifically for hybrid algorithms. 
Equipped with the lifting theorem, we are able to prove directly NISQ security and complexity results by calculating a single combinatorial quantity, relying solely on classical reasoning.
 
As applications, we derive the first \emph{direct product theorems} in the average case, in 
the hybrid setting---i.e., an enabling tool to determine the hybrid hardness of solving multi-instance security games.
This allows us to derive in a straightforward manner the 
NISQ hardness of various security games, 
such as (i) the non-uniform hardness of salted games, 
(ii) the hardness of specific cryptographic tasks such as the multiple instance version of one-wayness and collision-resistance, and (iii) uniform or non-uniform hardness of many other games.

\end{abstract}

 \newpage
 \tableofcontents

\newpage

\section{Introduction}

Hash functions are a fundamental workhorse in modern cryptography. Efficient constructions such as SHA-2 and SHA-3 are widely used in real-world cryptographic applications. This is in no small part facilitated by the {\em random oracle model} (ROM) abstraction~\cite{bellare1993random}, which often enables efficient cryptographic constructions and simple analysis. 

The development of quantum computing poses significant threats to cryptography. \emph{Quantum random oracle model} (QROM) has since been proposed 
to take into account quantum attackers~\cite{AC:BDFLSZ11}.
The various techniques that have been developed for analyzing security in the QROM, however, 
are often either {\em ad-hoc} (for specific scenarios) or too involved to apply. They include approaches such as polynomial methods~\cite{BBC+01}, adversarial methods~\cite{Ambainis02}, small-range distributions~\cite{zhandry2012construct}, and compressed oracles~\cite{zhandry16record}.

Achieving post-quantum security often comes with the need for larger security parameters, hence causing an efficiency blow-up. For hash functions in particular, this leads to substantially longer output length and increased running time. On the other hand, the study of {\em Noisy Intermediate-Scale Quantum} (NISQ) devices~\cite{Preskill18}
constitutes a crucial intermediate goal. 
In the NISQ model, one focuses on: (i) \emph{hybrid} algorithms, where 
parties are granted a quota of both classical and quantum resources, resulting in a model which subsumes the purely classical or purely quantum computing models as special cases;
(ii) quantum algorithms with access to \emph{noisy oracles}, where each quantum query is affected by a dephasing noise; or
(iii) \emph{bounded-depth} quantum algorithms.
These models of computation are also practically motivated as a fully quantum algorithm typically requires running a large-scale quantum computer coherently for an extended amount of time, while in the near-to-intermediate term, the available quantum devices are likely to be computationally restricted as well as expensive~\cite{Preskill18}.
Understanding the complexity of various cryptographic tasks against adversaries with both classical and quantum capabilities is not only crucial in near-term cryptography, but also a key step towards advancing our knowledge of quantum computational complexity, and efforts have been made to analyze, in this model, pre-image resistance~\cite{Ros22,CGS23,HLS24} and collision resistance~\cite{HLS24}. However, while these methods establish tight security for these two cryptographic uses of hash functions, they appear too complex and specific to be 
ported to other applications.

\subsection{Summary of Our Results}

In this work we establish novel lifting theorems for search games in the presence of NISQ adversaries, 
specifically,
(i) hybrid algorithms that are equipped with both classical and quantum queries, (ii) quantum algorithms with access to noisy oracles, as well as (iii) quantum algorithms with bounded quantum depth. At the heart of this result is a novel measure-and-reprogram technique, which we refer to as \emph{hybrid coherent reprogramming}. This technique can be viewed as a generalization of the coherent reprogramming method \cite{previous_work}.

The lifting theorems allow us to prove quantum and NISQ security for various cryptographic applications by calculating a single combinatorial quantity, relying solely on classical reasoning. Using this new framework, we derive the first \emph{direct product theorems} in the average case, 
in the NISQ setting.

\subsubsection*{Hybrid Lifting Theorem for Search Games.}
 Our central result 
  proposes a novel \emph{hybrid lifting theorem} for search games relating the success probability of an arbitrary hybrid algorithm to the success probability of a hybrid algorithm performing a small number of queries to the RO. 

\begin{theorem}[Hybrid Lifting Theorem -- Informal]
Let $\gamemath$ be a search game with a classical challenger $\cC$ that performs at most $k$ queries to the RO, and let $\cA$ be a hybrid adversary equipped with $q$ quantum and $c$ classical queries in the game $\gamemath$ (against the $k$-classical query challenger $\cC$).
Then there exists a $k$ (quantum + classical) hybrid query adversary\footnote{{Algorithm $\cB$ has oracle access to the random oracle as well as to the 
algorithm $\cA$.}} $\cB$ such that:
\begin{equation*}   
\Pr[\cB \textrm{ wins } \gamemath] \geq 
\frac{1}{\left(O\left( \frac{q^2}{k^2} + \frac{c}{k} \right)\right)^{k}} 
\Pr[\cA \text{ wins } \gamemath],
\end{equation*}
{where the probabilities are taken over the 
RO and 
the randomness of 
algorithms $\cA$ and $\cB$,
respectively.}
\end{theorem}

\subsubsection*{Lifting Theorem for Search Games in the Noisy Oracle Model.}
Next, we give the lifting theorem in the \emph{noisy oracle} setting. In this model, we assume a quantum algorithm that performs quantum queries to an oracle $O_p$ where,
with probability $p$, the input/output registers will be measured, and the query becomes a classical query, and with
probability $1 - p$, it will be a regular quantum query.

\begin{theorem}[Lifting Theorem for Noisy Oracles -- Informal] 
Let $\gamemath$ be a search game with a classical challenger $\cC$ that performs at most $k$ queries to the noisy oracle $O_p$, and let $\cA$ be a quantum adversary equipped with $T$ queries in the game $\gamemath$ (against the $k$-classical query challenger $\cC$).
Then there exists a $k$-query adversary $\cB$ such that:
\begin{equation*}   
\Pr[\cB \textrm{ wins } \gamemath] \geq 
 E(p, T, k) 
        \cdot
\Pr[\cA \text{ wins } \gamemath].
\end{equation*}
  Where:
    \begin{align*}
        E(p, T, k) = \frac{1}{O\left( {T \choose k} \left( ((1-p) T k)^k + k \right) \right)}. 
    \end{align*}    
\end{theorem}

\begin{remark}
The bound is tight in the regime where \( p \) is close to \( 0 \) or \( 1 \). For instance, let's consider the function inversion problem where the goal is to find $x$ such that $f(x) = 0$; in this case $k = 1$.
When \( p \) approaches \( 0 \), it recovers the tight bound for a \( T \)-query quantum algorithm; conversely, as \( p \) nears \( 1 \), the bound degrades to that of a \( T \)-query classical algorithm.

However, this bound is not generally tight when $p$ is away from $0$ and $1$. Specifically, for function inversion, as shown in~\cite{HLS24}, when \( p \) is a constant, the best achievable performance of a $T$-quantum algorithm using a noisy oracle \( O_p \) is no better than that of a classical algorithm (i.e., having advantage $O(T/N)$), whereas our result gives $O(T^2/N)$. 

We note that establishing a general lifting lemma that tightly applies to noisy oracles or bounded-depth algorithms appears unattainable, as there exist quantum algorithms with only constant quantum depth that still achieve quantum advantage. For example, the algorithm in \cite{yamakawa2022verifiable} demonstrates a separation between quantum and classical algorithms while requiring only depth-1 quantum computation. Assume that if even allowing adaptive queries (quantum depth more than $1$), the algorithm in~\cite{yamakawa2022verifiable} is still optimal; then this suggests that constant-depth quantum algorithms in search games within the QROM can attain the maximum possible quantum advantage. Consequently, a general theorem establishing a separation between bounded and unbounded depth, independent of the specific problem, may not exist. Another example is the oracle interrogation problem~\cite{van1998quantum}, where the optimal quantum algorithm only has quantum depth $1$. 

Due to the inherent non-tightness, we do not explicitly state the bounds for all applications considered later in this section within the noisy oracle or bounded-depth models. However, deriving these lower bounds is just as straightforward as in the hybrid query model.
\end{remark}

\begin{corollary}[Lifting for Bounded-Depth] 
Let $\gamemath$ be a search game with a classical challenger $\cC$ that performs at most $k$ queries. Let $\cA$ be any $d$-depth bounded algorithm performing $T$ quantum queries in total in the game $\gamemath$ (against the $k$-classical query challenger $\cC$).
 Then there exists a $k$ query adversary $\cB$ against the game such that:
 \begin{equation*}
 \Pr[\cB^{H} \text{ wins } \gamemath] \geq 
 E(1/d, 2T, k)
 \Pr[\cA^{H} \text{ wins } \gamemath].  
 \end{equation*}
\end{corollary}

\subsubsection*{Hybrid Coherent Reprogramming.}
At the core of our main lifting results lies a new framework for quantum reprogramming, which we call \emph{hybrid coherent reprogramming}. 
This novel framework is the central piece for proving the hybrid lifting result, as well as several applications in NISQ security and complexity.

In order to present our main coherent reprogramming result, we first need to introduce a few notions. For an oracle $H$, we call $H_{x, y}$ the reprogrammed oracle that behaves almost like the original function $H$, with the only difference that its value on input $x$ will be $y$. Similarly, we define the reprogrammed oracle on $k$ inputs $\vec{x} = (x_1, ..., x_k)$ and $k$ corresponding outputs $\vec{y} = (y_1, ..., y_k)$, denoted by $H_{\vec{x}, \vec{y}}$, as the original function $H$ with the only difference that for every input $x_i$ in $\vec{x}$, the corresponding image will be $y_i$ in $\vec{y}$.

\begin{theorem} [Hybrid Coherent Reprogramming -- Informal] \label{thm:hybrid_coherent_reprogram_informal}
    Let $H, G$ be two random oracles. 
    Let $\As$ be any hybrid algorithm equipped with $q$ quantum and $c$ classical queries to the oracle $H$, and let $\xout = (x_1, ..., x_k) \in X^k$ be any $k$-vector of inputs and $\yout = (y_1, ..., y_k) = (G(x_1), ..., G(x_k))$.
    Then there exists a simulator algorithm $\Sim$ that given oracle access to $H, G$ and $\cA$, simulates the output of $\cA$ having oracle access to $H_{\xout, \yout}$ (the  reprogrammed version of $H$) with probability: 
      \begin{align*}
     & \Pr_{H, G}\left[\Sim \text{ outputs correct } (\vec{x}, \vec{y})  \right] \\
       \geq &  \frac{1}{\left(O\left( \frac{q^2}{k^2} + \frac{c}{k} \right)\right)^{k}} 
  \cdot  \Pr_{H, G}\left[\cA^{H_{\xout, \yout}} \text{ outputs correct } (\vec{x}, \vec{y}) \right]. 
    \end{align*}
Furthermore, $\Sim$ makes exactly $k$ (quantum+classical) queries to $G$.\footnote{{Here "correct" output is defined with respect to an arbitrary predicate that has access to the reprogrammed oracle.}}
\end{theorem}

Next, we show the applications of our hybrid lifting theorems in NISQ query complexity and cryptography.

\subsubsection*{Hybrid Lifting Theorem with Classical Reasoning.}
A \multigame between a challenger and an adversary is defined as follows. The adversary receives $k$ different challenges from the challenger, and at the end of their interaction, the adversary needs to respond with $k$ outputs. If the $k$ outputs (taken together) satisfy some relation $R$ specified by the game, we say the adversary wins the multi-output $k$-search game.
The goal of the \emph{hybrid lifting theorem} is to establish the hardness of solving the \multigame by any general hybrid adversary, with only simple classical reasoning. For an arbitrary $k$-ary relation $R$, let $\sym_k$ be the symmetric group on $[k]$ and 
define: 
\begin{align*}
p(R): = \Pr[\exists \pi\in \sym_k \ | \ (y_{\pi(1)}, y_{\pi(2)}, ..., y_{\pi(k)}) \in R  : (y_1, ..., y_k) \xleftarrow{\$} Y^k] \,.
\end{align*}
Note that $p(R)$ is a quantity that only depends on the game itself, and can be calculated with only classical reasoning. 

\begin{theorem}[Hybrid Lifting Theorem with Classical Reasoning -- Informal]
    For any hybrid algorithm $\cA$ equipped with $q$ quantum and $c$ classical queries to a random oracle $H : X \rightarrow Y$, $\cA$'s success probability to solve the \multigame as specified by the winning relation $R$, is bounded by:
    \begin{align*}
        \Pr[\cA \text{ wins \multigame}] \leq \left(O\left( \frac{q^2}{k^2} + \frac{c}{k} \right)\right)^{k} \cdot p(R) \, .
    \end{align*}
\end{theorem}

This lifting {theorem translates} into the following NISQ (hybrid and noisy) hardness results for our applications in query complexity and cryptography. 

\subsubsection*{Direct Product Theorem in the Hybrid Setting.} 
We give the first
\emph {hybrid
direct product theorem (DPT) in the average case} (in the QROM). Let $\mathcal{G}$ be a game defined in the QROM. Then $\mathcal{G}^{\otimes g}$ is defined as 
follows: An algorithm gets oracle access to $g$ independent random oracles (and can in superposition query all oracles at the same time); the goal is to find a valid output under each of these $g$ random oracles. 
Previously, only worst-case quantum DPTs were known \cite{sherstov2011strong,lee2013strong} and were proof-method dependent; 
recently, ~\cite{dong2024salting} and \cite{previous_work} showed the first average-case quantum DPTs for some problems in the QROM, 
while 
remaining proof-method dependent. In contrast, while our DPTs are non-tight, they work for hybrid quantum algorithms and are proof-method independent.

Concretely, our hybrid direct product theorem establishes the hardness of solving $g$ independent instances 
of a game $\gamemath$ given a total of $g \cdot q$ quantum queries and $g \cdot c$ classical queries:

\begin{theorem}[Hybrid Direct Product Theorem]
    For any hybrid algorithm $\As$ equipped with $g \cdot q$ quantum queries and $g \cdot q$ classical queries, $\As$'s success probability to solve the Direct Product game $\gamemath^{\otimes g}$ with the underlying $\gamemath$ specified by the winning relation $R$, is bounded by:
\begin{align*}
    \Pr[\As & \text{ wins } \gamemath^{\otimes g}] \leq 
     \left( \left(O\left( \frac{q^2}{k^2} + \frac{c}{k} \right)\right)^{k} p(R) \right)^g.
\end{align*}
\end{theorem}

\subsubsection*{Non-uniform Security of Salting in the Hybrid Setting.} 
The above 
{theorem directly implies} the
non-uniform security of salting. 
Non-uniform attacks allow for heavy-duty offline computation. 
Salting is a generic method that prevents non-uniform attacks against hash functions. Chung {\em et al.}~\cite{chung2020tight} show that ``salting generically defeats quantum preprocessing attacks''; they show that if a game in the QROM is $\epsilon(q)$ secure, the salted game with salt space $[K]$ is $\epsilon(q) + \frac{Sq}{K}$ secure against a quantum adversary with $S$-bit of advice. Their bound is non-tight, since when the underlying game is collision-finding, the tight non-uniform security should be $\epsilon(q) + \frac{S}{K}$. Improving the additive factor is an interesting open question,
which recently Dong {\em et al.}~\cite{dong2024salting} were able to answer 
affirmatively but for a limited collection of games. 

\begin{theorem}
For any non-uniform quantum algorithm $\As$ equipped with $q$ quantum and $c$ classical queries and $S$-bit of classical advice, $\As$'s success probability to solve the salted game $\gamemath_s$ with the underlying $\gamemath$ specified by a winning relation $R$, is bounded by:
\begin{align*}
    \Pr[\As & \text{ wins } \gamemath_s] \leq 
      4 \cdot \frac{S}{K} +  {\left(O\left( \frac{q^2}{k^2} + \frac{c}{k} \right)\right)^{k}}  p(R) .
\end{align*}
\end{theorem}

\subsubsection*{Non-Uniform Security.}
By combining our lifting theorem with the results by Chung {\em et al.}~\cite{chung2020tight}, we derive the following results concerning the security (hardness) against non-uniform hybrid adversaries with classical advice, for a broader class of games.

\begin{lemma}[Security against Hybrid Non-Uniform Adversaries (Informal)]
Let $\gamemath$ be any classically verifiable search game specified by the winning relation $R$. Let $R^{\otimes S}$ be the winning relation of the multi-instance game of $\gamemath$. 
Any hybrid non-uniform algorithm $\cA$ equipped with $q$ quantum and $c$ classical queries and $S$ classical bits of advice can win $\gamemath$ with probability at most:
\begin{align*}
   \Pr[\cA \text{ wins } \gamemath] \leq \left(O\left( \frac{S^2 q^2}{k^2} + \frac{S c}{k} \right)\right)^{\frac{k}{S}}
     \cdot  p(R^{\otimes S})^{1/S} \, .
\end{align*}
\end{lemma}

\medskip
Next,
to demonstrate the power of our results, we also apply them to derive
the NISQ security of \emph{cryptographic tasks} and the \emph{hybrid query complexity}.
Note that the list of applications we 
give below is non-exhaustive, given that $p(R)$ is easy to define for almost every game.

\medskip
\noindent{\bf Hardness and Optimality of Multi-Image Inversion in the Hybrid Setting.}
Firstly, we can analyze the hybrid hardness of inverting $k$ different images of a random oracle $H : [M] \rightarrow [N]$.
For this problem, we can show that the hybrid coherent measure-and-reprogram gives tight bounds.

\begin{lemma}[Hybrid Hardness of Multi-Image Inversion (Informal)] \label{lemma:hybrid_multi_image_informal}
    Any hybrid algorithm equipped with $q$ quantum queries and $c$ classical queries can solve the multi-image search problem with probability at most:
      \begin{align*}
     P_{max} = {\left(O\left( \frac{q^2}{k^2} + \frac{c}{k} \right)\right)^{k}}
 \cdot \frac{k!}{N^k} 
    \end{align*}
\end{lemma}
{
\begin{lemma}[Optimality of Hybrid Reprogramming for Multi-Image Search]
    The hybrid coherent measure-and-reprogram bound (\Cref{thm:hybrid_coherent_reprogram_informal}) is optimal for the multi-image inversion problem. Namely, there exists a  hybrid algorithm $\cA$ whose success probability matches up to constant factors the hardness bound in \Cref{lemma:hybrid_multi_image_informal}.
\end{lemma}
}

\medskip
\noindent{\bf Hardness of Multi-Collision Finding.}  Secondly, we can analyze the hybrid hardness of finding $k$ collisions, namely, $k$ inputs that map to the same image of the random oracle $H : [M] \rightarrow [N]$. 

\begin{lemma}[Hybrid Hardness of Multi-Collision Finding (Informal)]
    For any hybrid algorithm $\cA$, equipped with $q$ quantum and $c$ classical queries,   
    the probability of solving the $k$-multi-collision problem is at most:
    \begin{equation*}
        \Pr_H[\cA \rightarrow \vec{x} = (x_1, ..., x_k) \ : \ H(x_1) = ... = H(x_k)] \leq  {\left(O\left( \frac{q^2}{k^2} + \frac{c}{k} \right)\right)^{k}}
 \cdot \frac{1}{N^{k - 1}} .
    \end{equation*}
\end{lemma}

\medskip
\noindent{\bf Hardness of Multi-Search.}
Finally, we also establish a tight bound for finding $k$ distinct inputs that all map to $0$ under the random oracle $H : [M] \rightarrow [N]$. This is potentially useful in analyzing proofs of work in the blockchain context (cf.~\cite{garay2015bitcoin}).

\begin{lemma}[Hybrid Hardness of Multi-Search]\
For any hybrid algorithm $\cA$ equipped with $q$ quantum and $c$ classical queries, whose task is to find $k$ different preimages of $0$ of the random oracle, the success probability of $\cA$ is upper bounded by:
    \begin{align*}
          \Pr_H[\cA \rightarrow \vec{x} = (x_1, ..., x_k) \ : \ H(x_i) = 0 \ \forall i \in [k]] 
             \leq 
             {\left(O\left( \frac{q^2}{k^2} + \frac{c}{k} \right)\right)^{k}}
             \cdot \frac{1}{N^k}.
    \end{align*}
\end{lemma}

\medskip
\noindent{\bf Hardness of the 3SUM problem.} Recall that in this task, an algorithm needs to find 3 inputs
that are mapped by a random oracle $H: [M] \rightarrow \{-N, -N+1, ..., 0, ..., N\} $ to 3 images
that must sum up to $0$.

\begin{lemma}[Hybrid Hardness of 3SUM]
       For any hybrid algorithm $\cA$ equipped with $q$ quantum and $c$ classical queries, the success probability of $\cA$ to solve the 3SUM problem is upper bounded by:
         \begin{align*}
    \Pr_H[\cA \rightarrow \vec{x} = (x_1, x_2, x_3) \ : \ H(x_1) + H(x_2) + H(x_3) = 0] 
             \leq O \left(\left( q^2 + c \right)^{3} \cdot 
             \frac{1}{N} \right)
    \end{align*}
\end{lemma}

\subsection{Related Work}

The measure-and-reprogram framework was proposed and subsequently generalized and improved with tighter bounds, in the works of \cite{DFM20,LZ19,DFMS22}. A main application of the framework has been in the context of the Fiat-Shamir transformation, 
with several works establishing its post-quantum security 
\cite{Chailloux19,DFMS19,GHHM21}. Other cryptographic applications of measure-and-reprogram  have been considered in \cite{Katsumata21,BKS21,ABKK23,JMZ23,JSYP23,KX24}. Applications in query complexity of the framework have been developed in \cite{chung2020tight,YZ21}.

The {\em coherent} measure-and-reprogram framework for the purely quantum setting was recently proposed in~\cite{previous_work} 
as an improvement over the original measure-and-reprogram framework, yielding tighter bounds and implying lifting theorems in the 
QROM.

In the hybrid algorithm setting, to our knowledge, previous works have only focused 
on establishing hardness results for very specific search problems.
Rosmanis studied
the basic unstructured
search problem in the hybrid query model~\cite{Ros22} in order to find the unique solution of the problem. This hardness bound is also shown in~\cite{HLS24}, by a new recording technique tailored to the hybrid query model.
The work of \cite{CGS23} generalizes the Rosmanis result 
by studying the hybrid hardness for search problems where the input is sampled according to
arbitrary distributions. 
Finally, the work of~\cite{HLS24} proves the hardness of the collision problem by their 
generalized recording technique.

{
\medskip\noindent {\bf Comparison with \cite{previous_work}.}
It is worth contrasting the contributions of our work to a prior work in~\cite{previous_work}.
First off, instead of the coherent measure-and-reprogram approach \cite{previous_work},
our starting point is the \emph{standard} measure-and-reprogram technique~\cite{DFMS22}.
Indeed, our core idea can be used to make both measure-and-reprogram techniques work in the hybrid and NISQ models. While basing on the coherent version gives tighter results, 
without it, we can still get lifting lemmas for hybrid and NISQ algorithms based on \cite{DFMS22}. 

Conceptually, we believe our work 
changes how NISQ hybrid and noisy models are understood; for example,
the hybrid regime is not simply a convex combination of classical and quantum behaviors---subtle interference effects can arise when quantum queries are combined with classical ones, which our framework captures rigorously. As such, our result offers a clean and provably sound way to account for these effects in the ROM.

From a technical point of view, our new idea that enables the lifting lemmas for NISQ algorithms is in a sense ``orthogonal'' to that in \cite{previous_work}. 
Specifically, 
while our main theorem relies on a simulator that is adapted from the simulator in \cite{previous_work},
this adaptation can be applied to the simulator in the standard measure-and-reprogram lemma \cite{DFMS22},
which enables us to get a similar theorem, albeit with a worse loss compared to our main result. 
We choose to base our main result on \cite{previous_work} in order to obtain the tightest parameter values, but we emphasize that the underlying technique is independent of the specific measure-and-reprogram variant being used.

From an applications point of view, deriving NISQ security even for standard problems such as pre-image finding and collision-finding has been extremely challenging;
we believe that with our newly proposed framework one can derive in a straightforward, systematic and modular manner several NISQ security and complexity results.
}

\subsection{Organization of the Paper}

Our hybrid coherent measure-and-reprogram result is 
presented in \Cref{sec:hybrid_measure_reprogram},
and the lifting theorems for the noisy oracle and bounded-depth models are proven in \Cref{sec:lifting_noisy_bounded}. 
The complexity and cryptographic applications 
are presented as follows: hybrid lifting and direct product theorem tools are shown in \Cref{sec:applications},
and the non-uniform security, salted security and hardness of multi-image inversion, multi-collision finding and 3SUM 
are presented in \Cref{sec:other_app}.
Some deferred proofs can be found in \Cref{app:proofs}.

\section{Preliminaries}

\paragraph*{Notation.} For two vectors $\vec{x}, \vec{x}' \in X^k$, we say $\vec{x} \equiv \vec{x}'$ if and only if there exists a permutation $\sigma$ over the indices $\{1, 2, \ldots, k\}$ such that $x'_i = x_{\sigma(i)}$. For a function $H: X \to Y$ and $\vec{x} \in X^k$, $H(\vec{x})$ is defined as $(H(x_1), H(x_2), \ldots, H(x_k))$. We say $x \in \vec{x}$, if $x = x_i$ for some $i \in [k]$. For an ordered/unordered list $L = (x_1, y_1), (x_2, y_2), \ldots, (x_k, y_k)$, we say that $x \in L$ if and only if there exists $j \in [k]$ such that $x = x_k$. We can similarly define $x\not\in L$, $y \in L$ and $y \not\in L$ for every $x \in X, y \in Y$. 
For a vector $\vec{x}$, we denote by $|\vec{x}|$ the length of $\vec{x}$.

\subsection{Quantum Query Algorithms}
\label{sec:prelim_quantum_query_algo}

We will denote a quantum query algorithm by $\cA$. Let $q$ be the total number of quantum queries of $\cA$. By $\cA^H$ we mean that $\cA$ has quantum access to the function $H$. 

A quantum oracle query to $H$ will be applied as the unitary $O_H$: $O_H \ket{x} \ket{y} \longrightarrow \ket{x} \ket{y \oplus H(x)}$. 
Without loss of generality, we assume an algorithm will never perform any measurement (until the very end) and thus the internal state is always pure. 
We use $\ket {\phi_i^H}$ to denote the algorithm $\As$'s internal (pure) state right after the $i$-th query. 
\begin{align*}
    \ket {\phi^H_i} =  O_H U_i \cdots O_H U_1 \ket 0. 
\end{align*}
Specifically, we have, 
\begin{itemize}
    \item $\ket {\phi_0^H} = \ket 0$ is the initial state of $\As$; 
    \item $\ket {\phi_q^H}$ is the final state of $\As$. 
\end{itemize}
Without loss of generality, the algorithm will have three registers $\cX, \cY, \cZ$ at the end of the computation, where $\cX$ consists of a list of inputs, $\cY$ consists of a list of outputs corresponding to these inputs and some auxiliary information in $\cZ$. 

\paragraph*{Hybrid Algorithms.} We also consider hybrid query algorithms in this work. A hybrid query algorithm can make both quantum queries (as defined above) and classical queries to a random oracle, in an adaptive manner. We use $q$ to denote the number of quantum queries and $c$ to denote the number of classical queries. 

To make a classical query, an algorithm first measures its input register to obtain a classical input $x$. It then forwards $x$ to the classical interface of the oracle $H$, receiving the pair $(x, H(x))$ in return. For the remainder of this work, we will interpret a classical query in this equivalent manner:
\begin{itemize}
    \item Let $\sum_x \alpha_x \ket x$ be a state in the input register. Instead of measuring and querying it through the classical interface, the algorithm first makes a quantum query on $\sum_x \alpha_x \ket x \otimes \ket 0$ to obtain 
    \begin{align*}
        \sum_x \alpha_x \ket {x, H(x)}. 
    \end{align*}
    \item It prepares an empty register ${\cal H}$ (short for ``history register'') and CNOT both input and output into the register: 
    \begin{align*}
        \sum_x \alpha_x \ket {x, H(x)} \ket {x, H(x)}_{\cal H}.
    \end{align*}
    The algorithm will then never touch ${\cal H}$. 
\end{itemize}
The interpretation above serves as a purification of the classical query, eliminating the need for measurement. By adopting this perspective, we can assume that no measurements are made until the conclusion of the algorithm, ensuring that all internal states (e.g., $\ket {\phi_i^H}$) remain pure throughout.

\medskip

More formally, we will consider a new register --- the \emph{history} register --- denoted by $\cH$. The idea is that each classical query $x$ to the oracle $H$ will be appended at the end of the history register. 
 A classical oracle query to an oracle $H$ will behave as applying the unitary $O_H^{{C}}$:
\begin{equation*}
    O_H^{{C}} \ket{x}_{\cX} \ket{y}_{\cY} \otimes {\ket{C}_{\cH}} \longrightarrow \ket{x}_{\cX} \ket{y}_{\cY} \otimes {\ket{C \,||\, (x, H(x))}_{\cH}},
\end{equation*}
Here, ``$L\,||\,e$'' denotes appending an element $e$ to the end of the list $L$.

In the rest of the work, we will assume hybrid algorithms never make classical queries to the same input twice (\textbf{no duplicated classical queries}); since it can always store the input to the history register. Whenever it queries the same input classically, it will fetch the output from the history register and make a dummy query.

\paragraph*{Reprogramming Oracles.}
\begin{definition}[Reprogrammed Oracle] Reprogram oracle $H$ to output $y$ on input $x$, results in the new oracle, defined as:
$$
H_{x, y}(z) = 
\begin{cases}
     y,  & \text{ if } z = x \\
     H(z), & \text{ otherwise.}
\end{cases}
$$
We can similarly define a multi-input reprogram oracle $H_{\vec{x},\vec{y}}$ for $\vec{x} \in X^k$ without duplicate entries and $\vec{y} \in Y^k$:
$$
H_{\vec{x}, \vec{y}}(z) = 
\begin{cases}
     y_i,  & \text{ if } z = x_i \\
     H(z), & \text{ otherwise.}
\end{cases}
$$
\end{definition}

\subsection{The Quantum Measure-and-Reprogram Experiment}

We recall the measure-and-reprogram experiment and the state-of-the-art results here, first proposed by \cite{DFM20} and later adapted by \cite{YZ21}.  

\begin{definition}[Measure-and-Reprogram Experiment]\label{def:classical_measure_and_reprogram}
Let $\cA$ be a $q$-quantum query algorithm that outputs $\vec{x} \in X^k$ and $z \in Z$. For a function $H : X \rightarrow Y$ and $\vec{y} = (y_1, ..., y_k) \in Y^k$, define a measure-and-reprogram algorithm $\cB[H, \vec{y}]$:
\begin{enumerate}
    \item For each $j \in [k]$, uniformly pick $(i_j, b_j) \in ([q] \times \{0, 1\}) \cup \{(\perp, \perp)\}$ such that there does not exist $j \neq j'$ such that $i_j = i_{j'} \neq \perp$;
    \item Run $\cA^O$ where the oracle $O$ is initialized to be a quantumly accessible classical oracle that computes $H$ and when $\cA$ makes its $i$-th query, the oracle is simulated as follows:
    \begin{enumerate}
        \item If $i = i_j$ for some $j \in [k]$, measure $\cA$'s query register to obtain $x_j'$ and do either of the following:
        \begin{enumerate}
            \item If $b_j = 0$, reprogram $O$ using $(x_j', y_j)$ and answer $\cA$'s $i_j$-th query using the reprogrammed oracle;
            \item If $b_j = 1$, answer $\cA$'s $i_j$-th query using the oracle before reprogramming and then reprogram $O$ using $(x_j', y_j)$;
        \end{enumerate}
        \item Else, answer $\cA$'s $i$-th query by just using the oracle $O$ without any measurement or reprogramming;
    \end{enumerate}
    \item Let $(\vec{x} = (x_1, ..., x_k), z)$ be $\cA$'s output;
    \item For all $j \in [k]$ such that $i_j = \perp$, set $x_j' = x_j$
    \item Output $\vec{x'} := ((x_1', ..., x_k'), z)$
\end{enumerate}
\end{definition}

 We next state the state-of-the-art quantum measure-and-reprogram result. 

\begin{lemma}[Quantum Measure-and-Reprogram (adaptation from \cite{DFM20,YZ21})]
For any $H : X \rightarrow Y$, for any $\vec{x}^* = (x_1^*, ..., x_k^*) \in X^k$ without duplicated entries, for all $\vec{y} = (y_1, ..., y_k)$ and any relation $R \subseteq X^k \times Y^k \times Z$, we have:
\begin{equation*}
    \begin{split}
        \Pr&[\vec{x}' = \vec{x}^* \wedge (\vec{x}', \vec{y}, z) \in R \ | \ (\vec{x}', z) \leftarrow \cB[H, y]] \\ 
        & \geq \frac{1}{(2q+1)^{2k}} \Pr[\vec{x} = \vec{x}^* \wedge (\vec{x}, \vec{y}, z) \in R \ | \ (\vec{x}, z) \leftarrow \cA^{H_{\vec{x}^*, \vec{y}}}]
    \end{split}
\end{equation*}
    where $\cB[H, y]$ is the measure-and-reprogram experiment {and where the probabilities are taken over the randomness of $\cA$ and $\cB$,
    respectively.}
\end{lemma}

\subsection{Predicates and Success Probabilities \\}

\begin{definition}[Predicate/Verification Projection/Symmetric Predicate] 
Let $R$ be a relation on $X^k \times Y^k \times Z$. A predicate $V^H(\vec{x}, \vec{y}, z)$ parameterized by a random oracle $H$, returns $1$ if and only if $(\vec{x}, \vec{y}, z) \in R$ and $H(x_i) = y_i$ for every $i \in \{1, 2, \ldots, k\}$. 

Let $\cX, \cY, \cZ$ be the registers that store $\vec{x}, \vec{y}, z$, respectively. 
We define $\Pi^H_V$ as the projection corresponding to $V^H$: 
\begin{equation*}
    \Pi^H_V \ket {\vec{x}, \vec{y}, z} = 
    \begin{cases}
        \ket {\vec{x}, \vec{y}, z} & \text{ if } V^H(\vec{x}, \vec{y}, z) = 1;\\
        0 & \text{otherwise.}
    \end{cases}. 
\end{equation*}
\end{definition}

\medskip

Finally, for any predicate $V^H$, we are able to establish the success probability using the projection $\Pi^H_V$. \\
\begin{definition}[Success Probability] \label{def:succ_prob}
Let $\cA$ a quantum/hybrid query algorithm. Its success probability of outputting $\vec{x}, \vec{y}, z$ such that 
{$V^H(\vec{x}, \vec{y}, z) = 1$}
is defined as 
\begin{equation*}
    \Pr\left[\cA^H \rightarrow (\vec{x}, \vec{y}, z) \text{ and } V^H(\vec{x}, \vec{y}, z) = 1\right] = \norm{\Pi^H_V \ket{\phi_q^H}}^2, 
\end{equation*}
{where the probability is taken over the random oracle $H$ and over the randomness of $\cA$}
and where we recall that $\ket{\phi_q^H}$ refers to the final state of $\As$. \\

Sometimes, we care about the event that $\As$ outputs a particular $\vec{x}$ and still succeeds. 
For any $\xout$, the following probability denotes that $\As$ outputs $\vec{x} \equiv \xout$  and succeeds:
\begin{equation*}
    \Pr\left[\cA^H \rightarrow (\vec{x}, \vec{y}, z) ,\ \, \ \vec{x} \equiv \xout \,\text{ and } V^H(\vec{x}, \vec{y}, {z}) = 1 \right] = \norm{G_{\xout} \Pi^H_V \ket{\phi_q^H}}^2,
\end{equation*}
where $G_{\xout}$ is defined as the projection that checks whether $\cA$ outputs $\vec{x} \equiv \xout$. 
\end{definition}

\section{Hybrid Coherent Measure-and-Reprogram} \label{sec:hybrid_measure_reprogram}

In this section, we extend the
recently proposed framework~\cite{previous_work} to the hybrid setting, namely, 
the \emph{hybrid coherent measure-and-reprogram} framework that handles hybrid algorithms performing both classical and quantum queries. We will assume that the query pattern (i.e., whether each query is classical or quantum) has already been determined; that is, the algorithm cannot adaptively choose whether the next query will be
quantum or classical.
We will call such algorithms \emph{static algorithms}. In the next section, we recall a theorem from the literature showing that both are equivalent (up to a constant factor in the number of queries).

\subsection{Main Theorem}

\begin{theorem}[Hybrid Coherent Measure-and-Reprogram]
    \label{thm:hybrid_coherent_measure_and_reprogram}
    Let $H, G: X \to Y$ be two functions. 
    Let $k$ be a positive integer (which can be a computable function in both $n = \log{|X|}$ and $m = \log{|Y|}$). Let $V^H$ be any predicate defined over $X^k \times Y^k \times Z$.
 Let $\As$ be any static {hybrid} $q$-quantum and $c$-classical query  algorithm to the oracle $H$.
    Then there exists a black-box polynomial-time quantum algorithm $\Sim^{H, G, \As}$, satisfying the  properties below. 

For any $\xout \in X^k$ without duplicate entries and $\yout = G(\xout)$, we have, 
    \begin{align*}
        & \Pr\left[\Sim^{H, G, \As} \rightarrow (\vec{x}, \vec{y}, z)  \text{ and } \vec{x} \equiv \xout \text{ and } V^{H_{\xout, \yout}}(\vec{x}, \vec{y}, z) = 1\right]  \\ 
        \geq & 
        \frac{1}{
        {2^{2k} \cdot}
        k \cdot A_{k, q, c}}  \cdot \Pr\left[\cA^{H_{\xout, \yout}} \rightarrow (\vec{x}, z) \text{ and } \vec{x} \equiv \xout  \text{ and } V^{H_{\xout, \yout}}(\vec{x}, H_{\xout, \yout}(\vec{x}), z) = 1\right],
    \end{align*}
    where: 
    \begin{align*}
        A_{k, q, c} = \sum_{t = 0}^k {q \choose t}^2 \cdot {k \choose t} \cdot {c \choose {k - t}}. 
    \end{align*}
Alternatively, the bound can be simplified as follows:
    \begin{align*}
        & \Pr\left[\Sim^{H, G, \As} \rightarrow (\vec{x}, \vec{y}, z)  \text{ and } \vec{x} \equiv \xout \text{ and } V^{H_{\xout, \yout}}(\vec{x}, \vec{y}, z) = 1\right]  \\ 
        \geq & 
        \frac{1}{\left(O\left( \frac{q^2}{k^2} + \frac{c}{k} \right)\right)^{k}}  \cdot \Pr\left[\cA^{H_{\xout, \yout}} \rightarrow (\vec{x}, z) \text{ and } \vec{x} \equiv \xout  \text{ and } V^{H_{\xout, \yout}}(\vec{x}, H_{\xout, \yout}(\vec{x}), z) = 1\right],
    \end{align*}
    Furthermore, $\Sim$ makes exactly $k$ (quantum+classical) queries to $G$. 
\end{theorem}

\begin{remark}
By relying on the result 
by Don {\em et al.}~\cite{DFH22}, the 
theorem can be directly extended to any general hybrid algorithm in which the order of the classical and quantum queries can be adaptive (and can depend on the underlying oracle), at the cost of only a constant factor (specifically,
increasing the number of classical and quantum queries by a factor of 2).
\end{remark}

\medskip

\medskip

To show Theorem~\ref{thm:hybrid_coherent_measure_and_reprogram}, we will construct the simulator $\Sim^{H, G, \As}$ similar to the simulator in the purely quantum coherent measure and reprogram in \cite{previous_work}.
There is only one difference:
\begin{itemize}
    \item 
{Let $\vec{v}$ denote the indices of the queries to be reprogrammed.
Instead of picking them}
    uniformly at random, we will pick them according to a new distribution $a_t$, which we will specify 
    later. Our choice of $a_t$  leads to tight bounds for problems like function inversion using hybrid queries. 
\end{itemize}

Next, we formally state our hybrid experiment/simulator. 

\begin{definition}[Hybrid Coherent Measure-and-Reprogram Experiment]
\label{def:quantum_simulator_hybrid}
Let $\cA$ be a static hybrid algorithm equipped with 
$q$ quantum queries and $c + k$ classical queries that outputs $\vec{x} = (x_1, \ldots, x_k) \in X^k, \vec{y} \in Y^k$ and $z \in Z$.
We assume $\vec{y}$ is computed by $H(\vec{x})$, as the last $k$ queries. 

For a function $H : X \rightarrow Y$ and $G:  X \rightarrow Y$, define a measure-and-reprogram algorithm $\cB[H, G]$ as follows:
\begin{enumerate}
\setlength\itemsep{0.7em}
    \item Let $t$ represent the total number of quantum queries that are going to be reprogrammed out of the total $k$ number of reprogrammings. \ul{We will sample $t$ according to some distribution} $(a_t)_{0 \leq t \leq k}$ ($\sum_{t = 0}^k a_t = 1$), where $a_t$ denotes the probability of measuring-and-reprogram exactly $t$ quantum queries (and $k - t$ classical queries).
    \item Pick a subset $\vec{v}$ of $[c + q + k]$ of length $k$, \ul{uniformly at random, conditioned on having exactly $t$ values corresponding to quantum queries}. 
    We have $1 < v_1 < \cdots < v_k \leq c + q + k$. 
    Pick $\vec{b} \in \{0,1\}^k$ uniformly at random.

    \item Run $\cA$ with an additional control register ${\cal R}$, initialized as empty $\ket \emptyset$. Define the following operation $U$ that updates the control register: for $x$ that is not in $L$, 
    \begin{align*}
        U \ket x {\ket L}_{\cal R} \rightarrow \ket x {\ket {L \cup (x, G(x))}}_{\cal R}. 
    \end{align*}
       Here $L$ is the set of input and output pairs to be reprogrammed. Since we will only work with basis states $\ket x \ket L$ where $x$ is not in $L$, $U$ clearly can be implemented by
    a unitary (by assuming that the set $L$ is initialized as empty).
    
    \item When $\cA$ makes its $i$-th query,
    \begin{enumerate}
        \item If $i = v_j$ for some $j \in [k]$ and it is a quantum query, do either of the following:
        \begin{enumerate}
            \item If $b_j = 0$, {update ${\cal R}$ using the input register and $G$}, and make the $i$-th query to $H$ controlled by ${\cal R}$; 
            \item If $b_j = 1$, make the $i$-th query to $H$ controlled by ${\cal R}$ and {update ${\cal R}$ using the input register and $G$}.

        \item In the above procedure, before updating the control register, it checks coherently that the input register is not contained in the control register; otherwise, it aborts. Concretely, this corresponds to performing a projective measurement yielding a single bit of information, aborting if the input is contained, and otherwise proceeding with the rest of the simulation execution.
        \end{enumerate}
        
        \item Else, answer $\cA$'s $i$-th query (either quantum or classical) controlled by ${\cal R}$.
    \end{enumerate}
    \item Let $(\vec{x}, \vec{y}, z)$ be $\cA$'s output;
    \item Measure ${\cal R}$ register to obtain $L = (\vec{x}', \vec{y}')$. 
    \item Output $(\vec{x}, \vec{y}, z)$ if $\vec{x}' \equiv \vec{x}$; otherwise, abort. \\
\end{enumerate}
\end{definition}

\begin{remark}
    	We emphasize that the purpose of the simulator is to mimic $\cA$. Here we assume that the algorithm $\cA$, after the first $q$ queries, has already prepared the output $(x,z)$. We then force $A$ to make $k$ additional classical queries in order to generate $y = H(x)$.
		This helps simplifying the proof and since we can assume without loss of generality that $c+q \geq k$, this modification only results in at most a larger constant factor.
\end{remark}

\begin{remark}
    The distribution $(a_t)_{0 \leq t \leq q}$ of the $t$ locations to be reprogrammed 
    will be established later in the analysis (\Cref{subsec:hybrid_choice_a_t}). Intuitively, our objective 
    will be to determine what is the optimal value of $a_t$ that will minimize the overall loss (giving us the tightest reprogramming bound). \\ \\ \\
\end{remark}

\begin{proof}[Proof of \Cref{thm:hybrid_coherent_measure_and_reprogram}]
We sketch here the outline of the proof:
\begin{enumerate}
    \setlength\itemsep{0.5em}
    \item We first show how to bound the success probability of any hybrid algorithm $\cA$ by a linear combination of success probabilities of simulators (as defined in the hybrid coherent measure-and-reprogram experiment) that perform exactly $t$ quantum queries (\Cref{subsec:hybrid_bound_A}).
    This is the most technical part and at a high level it involves the following steps:
    \begin{itemize}[topsep=5pt]
        \setlength\itemsep{0.5em}
        \item Decompose the final state of $\cA$, after performing all the classical and quantum queries, based on when the measurements in the hybrid measure-and-reprogram experiment take place and whether the queries are made before or after the measure-and-reprogram operations;
        \item Take into account the final $k$ classical queries that will be performed by $\cA$ in order to generate the final output $(\vec{x}, H(\vec{x}))$;
        \item Further decompose the final state of $\cA$ based on how many of the measure-and-reprogram operations have occurred when $\cA$ was performing a quantum query. This allows upper bounding  the norm of the final state by a linear combination of un-normalized states corresponding to performing exactly $t$ reprogrammings of $\cA$'s quantum queries;
        \item Apply the verification projectors on each of the states in the decompositions results in a decomposition where we can fix the inputs appearing in the measure-and-reprogram of the classical and quantum queries;
        \item The tightness of the measure-and-reprogram bounds relies on the fact that the states corresponding to different classical queries to be reprogrammed will be orthogonal;
        \item Finally, express each term in this decomposition to a behavior of the simulator in the hybrid coherent measure-and-reprogram experiment.
    \end{itemize}
    \item In the second stage, we express the success probability of the simulator as a linear combination of the success probabilities of simulators winning with exactly $t$ quantum queries (\Cref{subsec:hybrid_reduce_B}).
    \item Next, we observe that the loss of the hybrid measure-and-reprogram comes from the simulator drawing in a suitable way the distribution for the indices of the queries that will be measured-and-reprogrammed. We show how to instantiate this distribution and then by putting all the pieces together (relating the success of $\cA$ with the success of the simulator) we can derive the hybrid coherent-and-reprogram \Cref{thm:hybrid_coherent_measure_and_reprogram}
(\Cref{subsec:hybrid_choice_a_t}).
    \item Finally, we show that the hybrid coherent measure-and-reprogram is optimal by studying the complexity of the multi-image search problem (\Cref{subsec:hybrid_optimal_multi_search}).
\end{enumerate} 
\end{proof}

\subsection{Bounding the Success of $\cA$}
\label{subsec:hybrid_bound_A}

To prove the central hybrid coherent measure-and-reprogram, we will first need to show the following result:
\begin{lemma} \label{lemma:hybrid_reduction_simulator_t}
Let $\As$ be any static {hybrid} $q$-quantum and $c$-classical query algorithm.
    Then Simulator $\cB$ defined in the hybrid coherent measure-and-reprogram experiment (\Cref{def:quantum_simulator_hybrid}) satisfies the following property for any $\xout \in X^k$ without duplicate entries and $\yout = G(\xout)$:
    \begin{align*}
        \Pr &\left[\cA^{H_{\xout, \yout}} \rightarrow (\vec{x}, z) \text{ and } \vec{x} \equiv \xout  \text{ and } V^{H_{\xout, \yout}}(\vec{x}, H_{\xout, \yout}(\vec{x}), z) = 1\right] 
        \leq \\ 
        &{2^k \cdot}
        k \cdot \sum_{t = 0}^{k} {q \choose t} {k \choose t} \sum_{\substack{\vec{v}, \vec{b} \\ |\vec{v}|_q = t}} p_{\xout, \vec{v},\vec{b}}, 
    \end{align*}
    where $|\vec{v}|_q = t$ represents that the simulator reprograms exactly $t$ quantum queries (i.e., $\vec{v}$ has exactly $t$ locations corresponding to quantum queries) and where
    $p_{\xout, \vec{v}, \vec{b}}$ denotes the probability that when the simulator $B$ picks $\vec{v}, \vec{b}$, it succeeds and  outputs $\vec{x} \equiv \xout$. 

\end{lemma}

\begin{proof}

We will denote the final state of a hybrid algorithm performing $c$ classical queries and $q$ quantum queries by $\ket{\phi_{c, q}}$.
Let $\vec{x} \in X^k$ be a fixed vector without duplicates and $\vec{y}$ a fixed vector in $Y^k$.
We can then describe the state of the hybrid algorithm as follows.

After making its $c$ classical and its $q$ quantum queries to $H_{\vec{x}, \vec{y}}$, the algorithm's state is
\begin{equation*}
    {
    \left|{\phi_{c, q}^{H_{\vec{x}, \vec{y}} }} \right\rangle =  O_{H_{\vec{x}, \vec{y}}}  U_{q + c} \cdots O_{H_{\vec{x}, \vec{y}}} U_1 \ket 0_{\As}. 
    }
\end{equation*}
We assume the algorithm begins in the state \( \ket{0} \), with a register that includes a classical history register, which we omit from the notation for now. There is no need to differentiate between classical and quantum queries, as the distinction only affects whether the input-output pair is added to the history register. Since the history register is not included in the notation, this distinction becomes relevant only in later analysis.

Next, we decompose this state of the hybrid algorithm $\cA$ so that each component corresponds to one of the cases in the Simulator (in \Cref{def:quantum_simulator_hybrid}). The decomposition step follows closely the idea from the purely quantum setting 
(proof of Theorem 6 in \cite{previous_work})
since we do not distinguish whether a query is classical or quantum at this step.

\removed{ 
\begin{align*}
  O_{H_{\vec{x}, \vec{y}}}^{{T_1}} U_1 \ket 0 
  &= O_H^{{T_1}} U_1 \ket 0 \vphantom{O_{H_{x_j, y_j}}^{{T_1}} \, \ketbra {x_j}{x_j}\,U_1 \ket 0} 
  - \sum_{x_j} O_H^{{T_1}} \ketbra {x_j}{x_j} U_1 \ket 0 \vphantom{O_{H_{x_j, y_j}}^{{T_1}} \, \ketbra {x_j}{x_j}\,U_1 \ket 0} + \sum_{x_j} O_{H_{x_j, y_j}}^{{T_1}} \, \ketbra {x_j}{x_j}\,U_1 \ket 0 \, \, \, \text{ -- State after 1st query} \\
  \left| {\phi^{H_{\vec{x}, \vec{y}}}_{c, q}} \right\rangle &= \sum_{\substack{\vec{v}, \vec{b}}} \ket {\phi_{\vec{v}, \vec{b}}},  \, \, \, \text{ -- State after all {$c + q$} queries, where} \\
&  \ket {\phi_{\vec{v}, \vec{0}}} = \sum_{\sigma \in S^k_t}  O_{H_{\vec{x}_\sigma, \vec{y}_\sigma}}^{{T_t}} U_{q} \cdots O_{H_{\vec{x}_\sigma, \vec{y}_\sigma}}^{{T_t}} U_{v_t+1} \underbrace{O_{H_{\vec{x}_\sigma, \vec{y}_\sigma}}^{{T_t}}\ketbra{x_{\sigma_t}}{x_{\sigma_t}} \cdots U_{v_{t-1}+1}}_{\text{stage (t)}} \\
    & \cdots \\
     & \cdot \underbrace{O_{H_{(x_{\sigma_1}, x_{\sigma_2}), (y_{\sigma_1}, y_{\sigma_2})}}^{{T_2}} \ketbra{x_{\sigma_2}}{x_{\sigma_2}} U_{v_2} \cdots O_{H_{x_{\sigma_1}, y_{\sigma_1}}}^{{T_2}} U_{v_1+1}}_{\text{stage (2)}}  \\
     & \cdot \underbrace{O_{H_{x_{\sigma_1}, y_{\sigma_1}}}^{{T_1}} \ketbra {x_{\sigma_1}}{x_{\sigma_1}}  U_{v_1} O_H^{{T_1}} \cdots O_H^{{T_1}} U_1 \ket 0 { \otimes \ket{\emptyset}_{\cH}} }_{\text{stage (1)}}
\end{align*}
}

The original state $\ket{\phi_{c, q}^{H_{\vec{x}, \vec{y}} }}$ was decomposed up to the first $q + c$ queries (instead of all $q+ c + k$ queries), 
in a summation of sub-normalized states $\phi_{\vec{v}, \vec{b}}$ parameterized by when the $\ell$ measurements happen, given by an ordered vector $\vec{v}$ such that: $ 1 \leq v_1 \leq ... \leq v_\ell \leq c + q $ and whether the queries are made before or after these measure-and-reprogram $b \in \{0, 1\}^\ell$, for $\ell \in \{0, ..., c + q\}$.

For each $\vec{v}$ of length $\ell$ and $b \in \{0, 1\}^\ell$, we define:
\begin{equation*}
    \ket{\phi_{\vec{v}, \vec{b}}} = \sum_{\sigma \in S_\ell^k} \ket{\phi_{\vec{v}, \vec{b}, \sigma}},
\end{equation*}
where $\ket{\phi_{\vec{v}, \vec{b}, \sigma}}$ is the state that is measured-and-reprogrammed according to the vectors $\vec{v}$ and $\vec{b}$ and the order $\sigma$ ($x_{\sigma_1}$ was first measured, and then $x_{\sigma_2}$, so on and so forth).

\removed{
We make the following observation regarding the hybrid algorithm's decomposition.
One can notice that the decomposition in the hybrid setting preserves the structure of the fully quantum setting decomposition. This is intuitively because in the hybrid setting the only difference is that when making a classical query we will add it to the separate History register. This in particular means that for the decomposition we can treat the classical query as a quantum query by merging the History register into quantum algorithm, resulting in a new quantum algorithm making $q + c$ quantum queries. 
}

\paragraph*{Adding the extra $k$ queries.}
Assume algorithm $\cA$ after making the first $q + c$ queries, already prepared output $\vec{x}, z$. We then force the algorithm to make its last $k$ queries in order to generate $\vec{y} = H(\vec{x})$.

Recall the definitions of $\Pi^{H_{\xout, \yout}}_V$ and $G_{\xout}$ in \Cref{def:succ_prob}. By setting $\vec{x} = \xout$ and $\vec{y} = \yout$ in the above analysis, the probability on the RHS of our \Cref{thm:hybrid_coherent_measure_and_reprogram} is equal to:
\begin{align*}
    \Pr &\left[\cA^{H_{\xout, \yout}} \rightarrow (\vec{x}, z) \text{ and } \vec{x} \equiv \xout  \text{ and } V^{H_{\xout, \yout}}(\vec{x}, H_{\xout, \yout}(\vec{x}), z) = 1\right] \\
    &= \left\|G_{\xout} \Pi_V^{H_{\xout, \yout}} \left| {\phi^{H_{\xout, \yout}}_{{c + k, q}}} \right\rangle \right\|^2.
\end{align*}

Since $G_{\xout}$ and $\Pi_V^{H_{\xout, \yout}}$ commute (as they both are projections over computational basis), we can assume $G_{\xout}$ is applied to the state first.
Even further, as the computation of $\vec{y} = H(\vec{x})$ and the projection $G_{\xout}$ also commute, we can assume $G_{\xout}$ applies to the state right before the last $k$ queries, which are used to compute $\vec{y}$.
Therefore, for every $\ket{\phi_{\vec{v}, \vec{b}, \sigma}}$, even if $\ell < k$ (the length of $\vec{v}$), we can measure-and-(immediately)-reprogram exactly $k - \ell$ locations of the last $k$ queries, and making the random oracle exactly reprogrammed to $H_{\xout, \yout}$.
Thus, we have: 
\begin{align}\label{eq:state_decompose_hybrid}
    \left| {\phi^{H_{\xout, \yout}}_{{c + k, q}}} \right\rangle = \sum_{\substack{\vec{v}, \vec{b} \\ |\vec{v}| = k}} \ket {\phi_{\vec{v}, \vec{b}}},
\end{align}
where the RHS has (at most) {$2^k \cdot \binom{q + c + k}{k}$} terms.

For a vector $\vec{v}$ of length $k$, define $\vec{v}_q$ be the subset of indices in $\vec{v}$ corresponding to quantum queries and by $|\vec{v}|_q$ the number of quantum queries in $\vec{v}$ (corresponding to quantum queries to be reprogrammed).
Then, we can further decompose the quantum state of the hybrid algorithm $\cA$, by first defining the following state for any $t \leq k$:
\begin{equation*}
    \ket {\phi_{t}} = \sum_{\substack{\vec{v}, \vec{b} \\ |\vec{v}|_q = t}} \ket{\phi_{\vec{v}, \vec{b}}}.
\end{equation*}
We will always assume $\vec{v}$ is a vector of length $k$ from now on, and simply ignore ``$|\vec{v}|=k''$ in the rest of the work. By applying the Cauchy-Schwarz inequality:
\begin{equation} \label{eq:state_t}
\left\|  \left| {\phi^{H_{\xout, \yout}}_{{c + k, q}}} \right\rangle \right\|^2 \leq k \cdot \sum_{t = 0}^{k} \left\| \ket {\phi_{t}} \right\|^2.
\end{equation}

We now further decompose the state with exactly $t$ quantum reprogramming based on all possible $t$ quantum indices to be reprogrammed. Namely, for a fixed vector $\vec{v}_{\sf qu}$ of length at most $k$, define the state:
\begin{equation*}
    \ket{\phi_{\vec{v}_{\sf qu}}} = \sum_{\substack{\vec{v}, \vec{b} \\ \vec{v}_q = \vec{v}_{\sf qu}}} \ket{\phi_{\vec{v}, \vec{b}}}.
\end{equation*}
Then, it becomes clear that we can express the state with $t$ reprogramming as:
\begin{equation*}
     \ket {\phi_{t}} = \sum_{\substack{\vec{v}_{\sf qu} \\ |\vec{v}_{\sf qu}| = t}} \ket{\phi_{\vec{v}_{\sf qu}}},
\end{equation*}
where the RHS has ${q \choose t}$ terms.
By applying again Cauchy-Schwarz:
\begin{equation*}
    \left\| \ket {\phi_{t}} \right\|^2 \leq {q \choose t} \sum_{\substack{\vec{v}_{\sf qu} \\ |\vec{v}_{\sf qu}| = t}} 
    \left\| \ket{\phi_{\vec{v}_{\sf qu}}} \right\|^2.
\end{equation*}
In particular, combining this with Equation~\ref{eq:state_t}, implies that the norm of the algorithm's state is bounded by:
\begin{equation} \label{eq:decomp_v_qu}
    \left\|  \left| {\phi^{H_{\xout, \yout}}_{{c + k, q}}} \right\rangle \right\|^2 \leq k \cdot \sum_{t = 0}^{k} {q \choose t} \sum_{\substack{\vec{v}_{\sf qu} \\ |\vec{v}_{\sf qu}| = t}} 
    \left\| \ket{\phi_{\vec{v}_{\sf qu}}} \right\|^2.
\end{equation}

When applying the two verification projections on each of the states in the decomposition, we know that every $x$ that gets outputted by the algorithm (one of the entries in $\vec{x}_{\sf o}$) will appear in one of the measure-and-reprogramming steps.
Then we can decompose the state after the projection, based on how many of these inputs appear in measure-and-reprogramming of quantum queries, denoted by $\vec{x}_{q}$, and how many of them appear in measure-and-reprogramming of classical queries, denoted by $\vec{x}_{c}$. Note that, if treated as sets (not ordered vectors), $\vec{x}_{\sf o}$ will be the union of $\vec{x}_{q}$ and $\vec{x}_{c}$.

\noindent Concretely, for a fixed vector $\vec{v}_{\sf qu}$ and for fixed disjoint sets $\vec{x}_{q}$, $\vec{x}_{c}$ such that their union is equal to the output $\vec{x}_{\sf o}$, define the state $\ket{\phi_{\vec{v}_{\sf qu}, \vec{x}_{q}, \vec{x}_{c}}}$ as the state that:
\begin{itemize}
    \item has $\vec{v}_{\sf qu}$ as the indices of the quantum queries reprogrammed;
    \item has all inputs appear in the measure-and-reprogramming of the quantum queries
    equal to $\vec{x}_{q}$;
    \item has all inputs appear in the measure-and-reprogramming of the classical queries
    equal to $\vec{x}_{c}$.
\end{itemize}
Then, for any $t \leq k$ it must hold that:
\begin{equation*}
    G_{\xout} \Pi_V^{H_{\xout, \yout}} \ket{\phi_{\vec{v}_{\sf qu}}} = \sum_{\substack{\vec{x}_{q}, \vec{x}_c \\  \vec{x}_{q} \cup \vec{x}_c = \vec{x}_{\sf o} \\  |\vec{x}_{q}| = t, \\ |\vec{x}_{c}| = k - t}}  G_{\xout} \Pi_V^{H_{\xout, \yout}} \ket{\phi_{\vec{v}_{\sf qu}, \vec{x}_{q}, \vec{x}_{c}}},
\end{equation*}
where the number of terms in the sum is ${k \choose t}$. Then by applying Cauchy-Schwarz:
\begin{equation} \label{eq:decomp_x_q}
     \left\| G_{\xout} \Pi_V^{H_{\xout, \yout}} \ket{\phi_{\vec{v}_{\sf qu}}} \right\|^2 \leq {k \choose t} \sum_{\substack{\vec{x}_{q}, \vec{x}_c \\ \vec{x}_q \cup \vec{x}_c = \vec{x}_{\sf o} \\ |\vec{x}_{q}| = t, \\ |\vec{x}_{c}| = k - t}} \left\| G_{\xout} \Pi_V^{H_{\xout, \yout}} \ket{\phi_{\vec{v}_{\sf qu}, \vec{x}_{q}, \vec{x}_{c}}} \right\|^2.
\end{equation}
For fixed vectors $\vec{v}_{\sf qu}$, $\vec{x}_{q}$ (and $\vec{x}_{c}$), we can now further decompose over the classical queries to be reprogrammed denoted by $\vec{v_{\sf cl}}$ (which essentially satisfies that $\vec{v} = \vec{v}_{\sf qu} \cup \vec{v_{\sf cl}}$):
\begin{equation} \label{eq:decomp_vc}
   G_{\xout} \Pi_V^{H_{\xout, \yout}} \ket{\phi_{\vec{v}_{\sf qu}, \vec{x}_{q}, \vec{x}_{c}}} = \sum_{\vec{v_{\sf cl}}} G_{\xout} \Pi_V^{H_{\xout, \yout}} \ket{\phi_{\vec{v}_{\sf qu}, \vec{v_{\sf cl}}, \vec{x}_{q}, \vec{x}_{c}}}.
\end{equation}
We fix $\vec{v}_{\sf qu}, \vec{x}_{q}, \vec{x}_{c}$. Then, we can show that 
for any $\vec{v_{\sf cl}} \neq \vec{v_{\sf cl}'}$: 
\begin{equation}
    G_{\xout}\Pi_V^{H_{\xout, \yout}} \ket{\phi_{\vec{v}_{\sf qu}, \vec{v_{\sf cl}}, \vec{x}_{q}, \vec{x}_{c}}} 
    \perp 
    G_{\xout}\Pi_V^{H_{\xout, \yout}} \ket{\phi_{\vec{v}_{\sf qu}, \vec{v_{\sf cl}'}, \vec{x}_{q}, \vec{x}_{c}}}.
\end{equation}
The orthogonality follows from the fact that if $\vec{v_{\sf cl}} \neq \vec{v_{\sf cl}'}$, this implies that there must exist some index $i$ corresponding to a classical query of $\cA$ that is reprogrammed in the first case ($\vec{v_{\sf cl}}$), but not in the second ($\vec{v_{\sf cl}'}$). Let us denote with $x_i$ this query. Since we assume the algorithm will only query each distinct input once, it implies $x_i$ is classically queried in the first case and thus is stored in this history register; whereas in the second case, it is not.
Therefore, we can argue that the orthogonality between the two states is due to the history register $\cal H$.

The orthogonality applied to the decomposition in Equation~\ref{eq:decomp_vc} implies that:
\begin{equation} \label{eq:orth_v_cl}
    \left\|G_{\xout} \Pi_V^{H_{\xout, \yout}} \ket{\phi_{\vec{v}_{\sf qu}, \vec{x}_{q}, \vec{x}_{c}}} \right\|^2  = \sum_{\vec{v_{\sf cl}}} \left\| G_{\xout} \Pi_V^{H_{\xout, \yout}} \ket{\phi_{\vec{v}_{\sf qu}, \vec{v_{\sf cl}}, \vec{x}_{q}, \vec{x}_{c}}}\right\|^2.
\end{equation}

By combining the above decompositions, namely  Equation~\ref{eq:decomp_v_qu}, Equation~\ref{eq:decomp_x_q} and Equation~\ref{eq:orth_v_cl}, we finally get:
\begin{align}
\label{eq:decomp_q_k_t}
\left\| G_{\xout} \Pi_V^{H_{\xout, \yout}}  \left| {\phi^{H_{\xout, \yout}}_{{c + k, q}}} \right\rangle \right\|^2 \nonumber
&\leq 
k \cdot \sum_{t = 0}^{k} {q \choose t} \sum_{\substack{\vec{v}_{\sf qu} : \\ |\vec{v}_{\sf qu}| = t}} 
    \left\| G_{\xout} \Pi_V^{H_{\xout, \yout}} \ket{\phi_{\vec{v}_{\sf qu}}} \right\|^2 \\
    & \leq 
    k \cdot \sum_{t = 0}^{k} {q \choose t} \sum_{\substack{\vec{v}_{\sf qu} : \\ |\vec{v}_{\sf qu}| = t}} 
    {k \choose t} \sum_{\substack{\vec{x}_{q} : \\ |\vec{x}_{q}| = t, \\ |\vec{x}_{c}| = k - t}} \left\| G_{\xout} \Pi_V^{H_{\xout, \yout}} \ket{\phi_{\vec{v}_{\sf qu}, \vec{x}_{q}, \vec{x}_{c}}} \right\|^2 \\
    &\leq 
    k \cdot \sum_{t = 0}^{k} {q \choose t} {k \choose t} \sum_{\substack{\vec{v}_{\sf qu} : \\ |\vec{v}_{\sf qu}| = t}} 
     \sum_{\substack{\vec{x}_{q} : \\ |\vec{x}_{q}| = t, \\ |\vec{x}_{c}| = k - t}} 
    \sum_{\vec{v_{\sf cl}}} \left\| G_{\xout} \Pi_V^{H_{\xout, \yout}} \ket{\phi_{\vec{v}_{\sf qu}, \vec{v_{\sf cl}}, \vec{x}_{q}, \vec{x}_{c}}}\right\|^2. \nonumber
\end{align}
We notice that going over all possible $\vec{v}_{\sf qu}$ (of size $t$), over all $\vec{v_{\sf cl}}$, and over all the possible vectors $\vec{x}_{q}$ and $\vec{x}_{c}$ is equivalent to going over the vector $\vec{v}$ of queries to be reprogrammed such that exactly $t$ quantum queries are reprogrammed. Hence, this is equivalent to:
\begin{align*}
\left\| G_{\xout} \Pi_V^{H_{\xout, \yout}}  \left| {\phi^{H_{\xout, \yout}}_{{c + k, q}}} \right\rangle \right\|^2
&\leq 
{2^k \cdot}
k \cdot \sum_{t = 0}^{k} {q \choose t} {k \choose t} \sum_{\substack{\vec{v}, \vec{b}: \\ |\vec{v}|_q = t}} \left\| G_{\xout} \Pi_V^{H_{\xout, \yout}}  \ket{\phi_{\vec{v}, \vec{b}}} \right\|^2.
\end{align*}

To prove the theorem statement, we relate each individual term on the RHS with the behaviors of the simulator $\cB$. 

\paragraph*{Relating each term with the simulator $\cB$.}
Next, we prove that each term 
$\left\| G_{\xout} \Pi_V^{H_{\xout, \yout}} \ket {\phi_{\vec{v}, \vec{b}}} \right\|^2$ is upper bounded by the probability that when the simulator $B$ picks $\vec{v}, \vec{b}$, it succeeds and  outputs $\vec{x} \equiv \xout$, which we denote by $p_{\xout, \vec{v}, \vec{b}}$. This follows exactly the same idea, as again we do not distinguish whether a query is classical or quantum at all. 
\removed{
Since simulator $\cB$ ensures that (1) no duplicated elements ever in the control register, (2) at the end, the control register only consists of inputs that are outputted by $\cA$ (which will be $\xout$, enforced by $G_{\xout}$), we have that $p_{\xout, \vec{v}, \vec{b}}$ is the squared norm of the state $\left(G_{\xout} \Pi_V^{H_{\xout, \yout}} \otimes I_{\mathcal{R}} \right) \ket {\psi_{\vec{v}, \vec{b}}}$, with the state $\ket {\psi_{\vec{v}, \vec{b}}}$ being: 
\begin{align*}
    \ket {\psi_{\vec{v}, \vec{b}}} = \sum_{\sigma \in S^k_k} \ket {\phi_{\vec{v}, \vec{b}, \sigma}} \otimes \left| {\sf set} \left\{(\xo{\sigma_1}, \yo{\sigma_1}), \ldots, (\xo{\sigma_k}, \yo{\sigma_k})\right\} \right\rangle_{\mathcal{R}}.
\end{align*}

The only difference between $\ket {\psi_{\vec{v}, \vec{b}}}$ and $\ket {\phi_{\vec{v}, \vec{b}}}$ is the extra control register! However, we realize that in this case, when $\sigma$ is a permutation of $[k]$,  the control register is unentangled, making $p_{\xout, \vec{v},\vec{b}}$ is equal to $\left\| G_{\xout} \Pi_V^{H_{\xout, \yout}} \ket {\phi_{\vec{v}, \vec{b}}} \right\|^2$. This is because the set will simply be ${\sf set}\{(\xo{1}, \yo{1}), \ldots, (\xo{k}, \yo{k})\}$, regardless of what $\sigma$ is. }
As a result:
\begin{align*}
     \Pr & \left[\cA^{H_{\xout, \yout}} \rightarrow (\vec{x}, z) \text{ and } \vec{x} \equiv \xout  \text{ and } V^{H_{\xout, \yout}}(\vec{x}, H_{\xout, \yout}(\vec{x}), z) = 1\right] = \left\| G_{\xout} \Pi_V^{H_{\xout, \yout}}  \left| {\phi^{H_{\xout, \yout}}_{{c + k, q}}} \right\rangle \right\|^2 \\
    & \leq
    {2^k \cdot}
    k \cdot \sum_{t = 0}^{k} {q \choose t} {k \choose t} \sum_{\substack{\vec{v}, \vec{b}: \\ |\vec{v}|_q = t}} \left\| G_{\xout} \Pi_V^{H_{\xout, \yout}}  \ket{\phi_{\vec{v}, \vec{b}}} \right\|^2 \\
     & \leq 
     {2^k \cdot}
     k \cdot \sum_{t = 0}^{k} {q \choose t} {k \choose t} \sum_{\substack{\vec{v}, \vec{b}: \\ |\vec{v}|_q = t}} p_{\xout, \vec{v},\vec{b}}.
\end{align*}
\end{proof}

\subsection{Reducing the Success of $\cB$}
\label{subsec:hybrid_reduce_B}

We are now going to express the probability of $\cB$ winning as a linear combination of the probability of winning with $t$ quantum queries. 
\begin{lemma}\label{lemma:sim_to_sim_t}
The Simulator $\cB$ defined in the hybrid coherent measure-and-reprogram experiment (\Cref{def:quantum_simulator_hybrid}) satisfies the following property:
\begin{align*}
    \Pr & \left[\Sim^{H, G, \As} \rightarrow (\vec{x}, \vec{y}, z)  \text{ and } \vec{x} \equiv \xout \text{ and } V^{H_{\xout, \yout}}(\vec{x}, \vec{y}, z) = 1\right] \\
    &= \frac{1}{2^k}  \sum_{t = 0}^{k}  \frac{1}{{q \choose t} \cdot {c \choose {k - t}}} \cdot a_t \cdot  \sum_{\substack{\vec{v}, \vec{b}: \\ |\vec{v}| = k, \\ |\vec{v}|_q = t}} p_{\xout, \vec{v},\vec{b}}.
\end{align*}
\end{lemma}

\begin{proof}
From the definition of $p_{\xout, \vec{v},\vec{b}}$, it is clear that we can express the probability of success of the simulator $\cB$ as:
\begin{align*}
    \Pr & \left[\Sim^{H, G, \As} \rightarrow (\vec{x}, \vec{y}, z)  \text{ and } \vec{x} \equiv \xout \text{ and } V^{H_{\xout, \yout}}(\vec{x}, \vec{y}, z) = 1\right] \nonumber \\
    &= \sum_{t = 0}^{k} \sum_{\substack{\vec{v}, \vec{b}: \\ |\vec{v}| = k, \\ |\vec{v}|_q = t}} p_{\xout, \vec{v},\vec{b}} \cdot \Pr\left[\cB \text{ picks } (\vec{v}, \vec{b}) \right] \\
    &=  \sum_{t = 0}^{k}  \frac{1}{2^k \cdot {q \choose t} \cdot {c \choose {k - t}}} \cdot a_t \cdot  \sum_{\substack{\vec{v}, \vec{b}: \\ |\vec{v}| = k, \\ |\vec{v}|_q = t}} p_{\xout, \vec{v},\vec{b}},
\end{align*}
where $a_t$ is the probability defined in the experiment.
\end{proof}

\subsection{Relating the Simulator to $\As$ and Choice for $a_t$}
\label{subsec:hybrid_choice_a_t}

The loss of our measure-and-reprogramming comes from drawing in a suitable way the distribution referring to the indices for which we want to perform the $k$ measure-and-reprogramming.
The parameters we can tune are related to how we pick those indices to measure-and-reprogram --- given by $a_t$ --- specifying the probability of using $t$ quantum queries (and $k - t$ classical queries).

We propose the following distribution $a_t$. We will show later that this leads to tight query bounds on hybrid algorithms for function inversion problems (matching the best algorithm), and gives various lifting lemmas for hybrid query complexity.

\begin{definition}\label{def:choice_a_t}
    Consider the following choice for the distribution $(a_t)_{0 \leq t \leq k}$:
    \begin{align*}
        \alpha_t := {q \choose t}^2 \cdot {k \choose t} \cdot {c \choose {k - t}} / A_{k,q,c}, \ \ \ \ \ \  \text{where } \  
        A_{k,q,c} = \sum_{t = 0}^k {q \choose t}^2 \cdot {k \choose t} \cdot {c \choose {k - t}}. 
    \end{align*}
\end{definition}

By instantiating with this set of parameters $a_t := \alpha_t$, we are now ready to 
complete the full proof of
our main theorem (\Cref{thm:hybrid_coherent_measure_and_reprogram}).
\removed{
\begin{lemma}
    Let the Simulator $\cB$ defined in the hybrid coherent measure-and-reprogram experiment (\Cref{def:quantum_simulator_hybrid}) instantiated with the distribution $(a_t)_t$ in Definition~\ref{def:choice_a_t}.
    Then, we have the following hybrid coherent measure-and-reprogram result:
    \begin{align*}
        & \Pr_{H, G}\left[\Sim^{H, G, \As} \rightarrow (\vec{x}, \vec{y}, z)  \text{ and } \vec{x} \equiv \xout \text{ and } V^{H_{\xout, \yout}}(\vec{x}, \vec{y}, z) = 1\right]  \\ 
        \geq & \frac{1}{2^k \cdot A_{k, q, c} \cdot k}  \cdot \Pr_{H, G}\left[\cA^{H_{\xout, \yout}} \rightarrow (\vec{x}, z) \text{ and } \vec{x} \equiv \xout  \text{ and } V^{H_{\xout, \yout}}(\vec{x}, H_{\xout, \yout}(\vec{y}), z) = 1\right]. 
    \end{align*}
\end{lemma}
}
\begin{proof}
[Proof of \Cref{thm:hybrid_coherent_measure_and_reprogram}.]
By replacing $a_t$ in the statement of Lemma~\ref{lemma:sim_to_sim_t} with the definition above, we get:
\begin{align} \label{eq:sim_a_t}
        & \Pr  \left[\Sim^{H, G, \As} \rightarrow (\vec{x}, \vec{y}, z)  \text{ and } \vec{x} \equiv \xout \text{ and } V^{H_{\xout, \yout}}(\vec{x}, \vec{y}, z) = 1\right] \nonumber \\
        = & \frac{1}{2^k} \cdot \frac{1}{A_{k, q, c}}  \sum_{t = 0}^{k}
        {q \choose t} \cdot {k \choose t}
        \cdot  \sum_{\substack{\vec{v}, \vec{b}: \\ |\vec{v}| = k, \\ |\vec{v}|_q = t}} p_{\xout, \vec{v},\vec{b}}.
\end{align}

Finally, by combining Lemma~\ref{lemma:hybrid_reduction_simulator_t} with Equation~\ref{eq:sim_a_t} we get:
\begin{align*}
     \Pr & \left[\Sim^{H, G, \As} \rightarrow (\vec{x}, \vec{y}, z)  \text{ and } \vec{x} \equiv \xout \text{ and } V^{H_{\xout, \yout}}(\vec{x}, \vec{y}, z) = 1\right] \nonumber \\
     &=
        \frac{1}{
        {2^{2k} \cdot}
        A_{k, q, c} \cdot k} \cdot \left(
        {2^k \cdot}
        k \cdot \sum_{t = 0}^{k}
        {q \choose t} \cdot {k \choose t}
        \cdot  \sum_{\substack{\vec{v}, \vec{b}: \\ |\vec{v}| = k, \\ |\vec{v}|_q = t}} p_{\xout, \vec{v},\vec{b}} \right) \\
        & \geq 
        \frac{1}{
        {2^{2k} \cdot}
A_{k, q, c} \cdot k} \cdot \Pr_{H, G}\left[\cA^{H_{\xout, \yout}} \rightarrow (\vec{x}, z) \text{ and } \vec{x} \equiv \xout  \text{ and } V^{H_{\xout, \yout}}(\vec{x}, H_{\xout, \yout}(\vec{y}), z) = 1\right]. \nonumber
\end{align*}
\end{proof}

\subsection{Optimality for Multi-Image Search}
\label{subsec:hybrid_optimal_multi_search}

In this section, we will show that our hybrid reprogramming bound (Theorem~\ref{thm:hybrid_coherent_measure_and_reprogram}) is optimal for the multi-image search problem. Concretely, we propose a hybrid algorithm for the multi-image search problem and we show that its success probability matches up to constant factors the hardness bound derived from the hybrid reprogramming theorem.

\begin{corollary}[Hybrid Hardness of Multi-Image Search] \label{corr:hybrid_hardness_multi_image}
    Any hybrid algorithm equipped with $q$ quantum queries and $c$ classical queries can solve the multi-image search problem with probability at most:
    \begin{align*}
     P_{max} =  2^{2k} \cdot k \cdot A_{k, q, c} \cdot \frac{k!}{N^k} 
    \end{align*}
\end{corollary}
\begin{proof}
 This follows directly by plugging the hybrid reprogramming bound derived in this section (Theorem~\ref{thm:hybrid_coherent_measure_and_reprogram}) in the Lifting Theorem (Theorem~\ref{thm:hybrid_direct_product_strong}):
\begin{align}
    p(R) &= \frac{k!}{N^k} \nonumber \\
    P_{max} & = 2^{2k} \cdot k \cdot A_{k,  q, c} \cdot p(R) = 2^{2k} \cdot k \cdot A_{k, q, c} \cdot \frac{k!}{N^k} 
\end{align}
\end{proof}

\begin{lemma}\label{lemma:hybrid_alg_multi_image}
    There exists a hybrid algorithm $\cA$ equipped with $q$ quantum queries and $c$ classical queries that can solve the multi-image search problem with probability:
    \begin{align*}
        P_{\cA} 
        &= \frac{1}{2^k} \cdot k! \cdot \frac{1}{N^k} \cdot \left(v + u^2\right)^{k}
    \end{align*}
\end{lemma}

\begin{proof}
Consider the following hybrid algorithm $\cA$ equipped with $q$ quantum queries and $c$ classical queries. For simplicity, we will assume that $q$ and $c$ are multiples of $k$, {namely, that there exist positive integers $u$ and $v$ such that}
$q = k \cdot u$ and $c = k \cdot v$.
\begin{enumerate}
    \item In the first stage of the algorithm, $\cA$ searches for an $x$ that is a preimage for any of the $k$ images $Y = (y_1, ..., y_k)$. In this stage, $\cA$ spends $u$ quantum queries and $v$ classical queries;
    \item Suppose that $\cA$ succeeded in the first stage for an image $y_i$. Then in the second stage $\cA$ searches for an $x$ that is a preimage for any of the $k - 1$ images $Y - \{y_i\}$. In the second stage, $\cA$ also spends $u$ quantum queries and $v$ classical queries.
    \item We continue in the exact same manner for $k$ stages (in each of the stages, spending exactly $u$ quantum queries and $v$ classical queries).
    \end{enumerate}

First consider the success of a hybrid algorithm equipped with $c$ classical and $q$ quantum queries to find a preimage of an image $y$, for a search space of size $N$. For this problem, it is known \cite{CGS23} that the optimal success probability is at least $\frac{1}{2} \left(\frac{v}{N} + \frac{u^2}{N - v}\right) \geq \frac{1}{2N} \left( v + u^2 \right) $.

    Then, the success probability of $\cA$ is equal to:
        \begin{align}
        P_{\cA} &= \frac{k}{2N}\left(v + u^2\right) \cdot \frac{k - 1}{2N} \left(v + u^2\right) \cdot \cdots \cdot \frac{1}{2N} \left(v + u^2\right) \nonumber \\
        &= \frac{1}{2^k} \cdot k! \cdot \frac{1}{N^k} \cdot \left(v + u^2\right)^{k}
    \end{align}

\end{proof}

It is not hard to observe that for constant $k$, the success probability is up to a constant factor equal to $P_{max}$, hence our reprogramming bound is optimal for the constant $k$ regime.
Our goal is to also show that for any general $k$, our bound is optimal for the multi-image inversion problem.

\begin{lemma}[Optimality of Hybrid Reprogramming for Multi-Image Search] 
    \label{lemma:optimal_hybrid_multi_image}
    The hybrid coherent measure-and-reprogram bound (in Theorem~\ref{thm:hybrid_coherent_measure_and_reprogram}) is optimal for the multi-image inversion problem. Namely, the hybrid algorithm $\cA$ (in Lemma~\ref{lemma:hybrid_alg_multi_image}) matches up to constant factors the hardness bound in Corollary~\ref{corr:hybrid_hardness_multi_image}.
\end{lemma}

\begin{proof}

To show the optimality, we need to relate the following:
\begin{align*}
    2^{2k} \cdot k \cdot A_{k, q, c} \cdot \frac{k!}{N^k} & \text{ and }
    \frac{1}{2^k}\cdot k! \cdot \frac{1}{N^k} \cdot \left( \frac{c}{k} + \frac{q^2}{k^2} \right)^{k} \text{ , or equivalently: }\\
    A_{k, q, c} 
    & \text{ and }
    \frac{1}{k \cdot 2^{3k}} \cdot \left(\frac{c}{k} + \frac{q^2}{k^2} \right)^{k}
\end{align*}

Using the definition of $A_{k, q, c}$, to prove the optimality, we need to show:

\begin{align*}
    \sum_{t = 0}^k {q \choose t}^2 \cdot {k \choose t} \cdot {c \choose {k - t}} 
    \leq
    \left(O\left(\frac{c}{k} + \frac{q^2}{k^2}\right) \right)^{k}
\end{align*}
where $1 \leq k \leq \min(c, q)$.

Finally, the bound on $A_{k,q,c}$ is shown separately in Lemma~\ref{lemma:bound_a_k_q_c}.

\end{proof}

\begin{lemma} \label{lemma:bound_a_k_q_c}
We have for any $q, c > 0$ and $k < \min\{q,c\}$, the following upper bound on $A_{k, q, c}$:  
    \begin{align*}
       A_{k, q, c}:= \sum_{t=0}^k \binom{q}{t}^2 \binom{k}{t} \binom{c}{k-t} \leq \left(O\left( \frac{q^2}{k^2} + \frac{c}{k} \right)\right)^{k}.
    \end{align*}
\end{lemma}
\begin{proof}
    First, we start from the simple observation that since  $\binom{k}{t} \leq 2^k$:
    \begin{align*}
       A_{k, q, c} 
       \leq 
       2^k \sum_{t=0}^k \binom{q}{t}^2 \binom{c}{k-t}
    \end{align*}
    Let $t^* \in \{0, \cdots, k\}$, be defined as $t^* = {\sf argmax}_t \left\{ \binom{q}{t}^2 \binom{c}{k-t}\right\}$

Then, we have:
    \begin{align*}
        A_{k, q, c} 
       \leq 2^k \sum_{t=0}^k \binom{q}{t}^2 \binom{c}{k-t} \leq 2^k \cdot k \cdot  \binom{q}{t^*}^2 \binom{c}{k-t^*} \leq 2^{2k} \cdot  \binom{q}{t^*}^2 \binom{c}{k-t^*}.
    \end{align*}
    
    We now only need to bound $\binom{q}{t}^2 \binom{c}{k-t}$ for any $0 \leq t \leq k$. Using the inequality $\binom{n}{k} \leq \left(\frac{en}{k}\right)^k$ holding for any $k \leq n$, we have:
    \begin{align*}
        A_{k, q, c} 
       \leq 2^{2k} \cdot
       \binom{q}{t^*}^2 \binom{c}{k-t^*} & \leq 2^{2k} \cdot \left( \frac{e q}{t^*}\right)^{2t^*} \left( \frac{e c}{k-t^*}\right)^{k-t^*}  \\
       & \leq
       2^{2k} \cdot e^{k + t^*}
        \cdot \left( \frac{q^2}{{t^*}^2}\right)^{t^*} \left( \frac{c}{k-t^*}\right)^{k-t^*} \\
        & \leq
        {(2e)}^{2k} \cdot \left( \frac{q^2}{{t^*}^2}\right)^{t^*} \left( \frac{c}{k-t^*}\right)^{k-t^*}.
    \end{align*}

    Finally, we use AM-GM inequality: $x_1 \cdot x_2 \cdot \cdots \cdot x_n \leq \left(\frac{x_1 + \cdots x_n}{n}\right)^n$.

    As a result, by applying AM-GM to the $2k$ factors $q/t^*$ ($2t^*$ factors) and $\sqrt{c / (k - t^*)}$ ($2k - 2t^*$ factors) we have:
    \begin{align*}
        \left( \frac{q}{t^*}\right)^{2 t^*} \left( \sqrt{\frac{c}{k-t^*}}\right)^{2(k-t^*)} &\leq
        \left( \frac{\frac{q}{t^*} \cdot 2 t^* + \sqrt{\frac{c}{k-t^*}} \cdot 2 (k-t^*)}{2k}  \right)^{2k} \tag{AM-GM}\\
        &\leq \left( \frac{q}{k} + \sqrt{c \cdot \frac{k-t^*}{k^2}}  \right)^{2k}  \\
        &\leq \left( \frac{q}{k} + \sqrt{\frac{c}{k}}  \right)^{2k}   \tag{as $k - t^* \leq k$}.
    \end{align*}
Hence, we have:
\begin{align*}
    A_{k, q, c} 
       \leq 
       {(2e)}^{2k} \cdot
        \left( \frac{q}{k} + \sqrt{\frac{c}{k}}  \right)^{2k}
\end{align*}
Finally, using the simple inequality $(a + b)^2 \leq 2(a^2 + b^2)$, applied to $q/k$ and $\sqrt{c/k}$, we get:
\begin{align*}
    A_{k, q, c} 
       &\leq {(2e)}^{2k} \cdot \left(\left(\frac{q}{k} + \sqrt{\frac{c}{k}}\right)^2\right)^k \\
       &\leq {(2e)}^{2k} \cdot 2^k \cdot \left(\frac{q^2}{k^2} + \frac{c}{k} \right)^k \\
       &\leq
       \left(8e^2 \left(\frac{q^2}{k^2} + \frac{c}{k} \right)  \right)^k.
\end{align*}
\end{proof}

\section{Lifting Theorems in the Noisy Oracle and Bounded-Depth Models} \label{sec:lifting_noisy_bounded}

In this section, we  
give lifting theorems 
for NISQ algorithms.
We first describe the two NISQ models---noisy oracle and bounded-depth---and then show the main lifting results.

\subsection{The Noisy Oracle Setting}

\begin{definition}[Noisy Oracle]
   For an oracle $O: X \rightarrow Y$, we denote by $O_p$ the noisy oracle of $O$,  parameterized by a probability $p \in [0, 1]$ such that for an input quantum query $\ket{\psi} = \sum_x a_x \ket{x}$, $O_p$ will answer as follows:
    \begin{itemize}
        \item $O_p \ket{\psi} = \sum_x a_x \ket{x} \ket{O(x)}$. We call this a \emph{quantum answer} and happens with probability $1 - p$;
        \item Measure the input state $\ket{\psi}$ in the computational basis. Suppose the measurement outcome is $x$. Then output $O(x)$. We call this a \emph{classical answer} and this happens with probability $p$.
    \end{itemize}
\end{definition}

The main challenge in this setting, as opposed to the hybrid setting, is that in the noisy oracle setting the Simulator does not know if the oracle replies with a classical or a quantum answer to the quantum queries of $\cA$.

To model the algorithm having oracle access to the noisy oracle $O_p$, we will rely on the idea of \emph{purification of noisy queries}, used by Rosmanis in the context of analyzing quantum search in the presence of noise \cite{rosmanis23}.
Concretely, any query to noisy oracle $O_p$ can be thought of as 
first applying a  quantum query to the (noiseless) oracle,
receiving the answer $\rho$, and then
applying the map:
$$ M_p \rho = (1 - p) \rho + p M \rho,$$
where $M$ denotes measuring the state in the computational basis.
Hence, for each quantum query of the algorithm, we will first apply the oracle query $O_H$ followed by applying $M_p$. We will denote the composed operation by $O_{H, p} =  M_p (O_H \otimes I)$.

Then, any arbitrary algorithm can be denoted as a sequence of arbitrary unitaries $U_i$ (independent of the underlying oracle $H$) and operators $O_{H, p}$. As a result, the final state of a $q$-quantum query algorithm querying a noisy oracle $H$ will be: $\ket{\psi} = U_{t+1} O_{H, p}...U_2 O_{H, p} U_1 \ket{0}$.

\subsection{The Bounded-Depth Setting}

In this model, we 
assume that a $d$-depth quantum algorithm can act as an arbitrary quantum algorithm and then after every $d$ depth, the entire system of the algorithm will get measured. 

\begin{definition}[Bounded-Depth Quantum Algorithms]
Let $\cA$ be a quantum algorithm. Then we say that $\cA$ is $d$-depth bounded if after every $d$ queries, the entire system of $\cA$ gets measured by applying the following dephasing channel: if the state of $\cA$ after $d$ queries is $\rho$, then the state is mapped to:
$$\rho \rightarrow  \sum_{x, p, w} \ketbra{x, y, z}{x, y, z} \rho \ketbra{x, y, z}{x, y, z}$$
where $x$ refers to the input, $y$ to the output and $z$ to the work register of $\cA$.
\end{definition}

As shown by Hamoudi {\em et al.}~\cite{HLS24}, the two models above are related:

\begin{lemma}[\cite{HLS24}]\label{lemma:bounded_noisy}
Consider a $d$-depth quantum algorithm $\cA$ performing $T$ quantum queries. Then there exists an algorithm $\cB$ performing $2T$ queries to a noisy oracle $O_{1/d}$, (with the noise parameter $p = 1/d$), such that $\cB$ outputs the same outcome as $\cA$.
\end{lemma}

\subsection{Main Result}

The main result we will show in this section is a lifting theorem 
for algorithms having oracle access to noisy oracles. 
As a direct corollary, we can also immediately derive a similar lifting theorem for bounded-depth quantum algorithms.

\begin{theorem}[Lifting for Noisy Oracles] \label{thm:lifting_noisy}
Let $\gamemath$ be a multi-output $k$-search game. Let $\cA$ be any $T$-quantum query algorithm having oracle access to the noisy oracle $O_p$
in the game $\gamemath$ (against the $k$-classical query challenger $\cC$).
 Then there exists a $k$-query adversary $\cB$ against the game such that:
\begin{equation*}
\Pr[\cB^{H} \text{ wins } \gamemath] \geq 
 E(p, T, k) \cdot 
\Pr[\cA^{H} \text{ wins } \gamemath].   
\end{equation*}
    Where:
    \begin{align*}
        E(p, T, k) = \frac{1}{O\left( {T \choose k} \left( ((1-p) T k)^k + k \right) \right)}.  
    \end{align*} 
    {and where the probabilities are taken over the randomness of the oracles $H$ and $O_p$ and 
     over the randomness of the
    algorithms $\cA$ and $\cB$.}
\end{theorem}
This will be shown by relying on the coherent measure-and-reprogram for noisy oracles (\Cref{thm:noisy_coherent_measure_and_reprogram}) in the next section, following the exact same idea as in the proof of \Cref{lem:hybrid_lifting}.

\begin{corollary}[Lifting for Bounded Depth] \label{corr:lifting_bounded}
    Let $\gamemath$ be a multi-output $k$-search game. Let $\cA$ be any $d$-depth bounded algorithm performing $T$ quantum queries in total in the game $\gamemath$ (against the $k$-classical query challenger $\cC$).
 Then there exists a $k$-query adversary $\cB$ against the game such that:
 \begin{equation*}
 \Pr[\cB^{H} \text{ wins } \gamemath] \geq 
 E(1/d, 2T, k)
 \Pr[\cA^{H} \text{ wins } \gamemath].  
 \end{equation*}
\end{corollary}
This follows directly from \Cref{lemma:bounded_noisy}.

\subsection{Coherent Measure-and-Reprogram for Noisy Oracles}

\begin{theorem}[Coherent Measure-and-Reprogram for Noisy Oracles]
\label{thm:noisy_coherent_measure_and_reprogram}
    Let 
    $H, G: X \to Y$
    be two functions. 
    Let $k$ be a positive integer (can be a computable function in 
    both $n = \log{|X|}$ and $m = \log{|Y|}$).
    Let $V^H$ be any predicate defined over $X^k \times Y^k \times Z$. 
    Let $\As$ be any $T$-quantum query algorithm having oracle access to the noisy oracle $O_p$ computing the function $H$. 
    Then there exists a black-box quantum algorithm $\cB^{H, G, \As}$, satisfying the properties below.
    
    For any $\xout \in X^k$ without duplicate entries and $\yout = G(\xout)$, we have, 
    \begin{align*}
         \Pr &\left[\cB^{H, G, \As} \rightarrow (\vec{x}, \vec{y}, z)  \text{ and } \vec{x} \equiv \xout \text{ and } V^{H_{\xout, \yout}}(\vec{x}, \vec{y}, z) = 1\right]  \\ 
        \geq 
         &
        E(p, T, k) \cdot 
        \Pr\left[\cA^{H_{\xout, \yout}} \rightarrow (\vec{x}, z) \text{ and } \vec{x} \equiv \xout  \text{ and } V^{H_{\xout, \yout}}(\vec{x}, H_{\xout, \yout}(\vec{x}), z) = 1\right], 
            \end{align*}   
    where:
    \begin{align*}
        & E(p, T, k) = \frac{1}{O\left( {T \choose k} \left( ((1-p) T k)^k + k \right) \right)}. 
    \end{align*}    
    Furthermore, $\cB$ makes exactly $k$ quantum queries to the noisy oracle computing the function $G$ and has a running time polynomial in $n, m, k$ and the running time of $\As$.
\end{theorem}

\begin{proof}

The Simulator in the noisy oracle setting will be exactly the same simulator as in the hybrid coherent reprogramming experiment \Cref{def:quantum_simulator_hybrid}.

In this approach, 
for a $T$ query algorithm, 
we consider an oracle assignment referring to the types of queries performed to the oracle. This model, in particular, captures the noisy oracle as the oracle assignment is given by the noise distribution, i.e. for each of the $T$ queries the type of query is independently assigned to be either quantum with probability $1 - p$ or classical with probability $p$.
Now, we consider the behaviour of the algorithm $\cA$ and the measure-and-reprogram simulator $\cB$ under different oracle assignments.

Let $V := \{v_1, ..., v_k \}$ be the set of reprogrammed queries (by the measure-and-reprogram experiment).
Then, by the definition of simulator $\cB$, we have that the probabilities of $\cA$ and $\cB$ are equal if the oracle assignments only differ in the point to be reprogrammed, i.e., in any of $v_i \in V$. \\ \\

Next, we define the following random variables:
\begin{itemize}
    \item ${\sf type}$ -- denotes an oracle assignment (for the $T$ queries), i.e. for each of the $T$ queries specifying whether the query is classical or quantum.
    \item $q$ -- denotes the total number of quantum queries (out of $T$) under a given assignment ${\sf type}$.
    \item $q_{{\sf rep}}$ -- represents the number of quantum queries among the $k$ reprogrammed queries $V := \{v_1, ..., v_k \}$ under a given ${\sf type}$.
    \item ${\sf type}_V$ -- denotes the oracle assignments only for the queries to be reprogrammed in $V$ and ${\sf type}_{\bar{V}}$ denotes the oracle assignments for all the points that are not to be reprogrammed (not in $V$).
\end{itemize}

Then we can write the success probability of the $T$-query algorithm $\cA$ as:
\begin{align*}
    \Pr[\cA \text{ wins}] =  \sum_{{\sf type}} \Pr[{\sf type}] \cdot \Pr[\cA \text{ wins} \ | \ {\sf type}]
\end{align*}
On the other hand from \Cref{eq:decomp_q_k_t}, we know that:
\begin{align*}
    \Pr[\cA \text{ wins} \ | \ {\sf type}] \leq 
    k \cdot \sum_{v_1, ..., v_k} {q \choose q_{{\sf rep}}} \cdot {k \choose q_{{\sf rep}}} \cdot \Pr[\cA \text{ wins } | \ v_1, ..., v_k, {\sf type}] 
\end{align*}
Therefore, we have:
\begin{align*}
    \Pr[\cA \text{ wins}] 
    &\leq k \cdot  \sum_{{\sf type}} \sum_{v_1, ..., v_k} \Pr[{\sf type}] \cdot {q \choose q_{{\sf rep}}} \cdot {k \choose q_{{\sf rep}}} \cdot \Pr[\cA \text{ wins } | \ v_1, ..., v_k, {\sf type}] \\
      \displaybreak[3]
    &\leq k \cdot  \sum_{{\sf type}} \sum_{v_1, ..., v_k} \Pr[{\sf type}] {T \choose q_{{\sf rep}}} \cdot {k \choose q_{{\sf rep}}} \cdot \Pr[\cA \text{ wins } | \ v_1, ..., v_k, {\sf type}] \\
    &\leq k \cdot \sum_{v_1, ..., v_k}  \sum_{{\sf type}_{\bar{V}}} \sum_{t = 0}^k \Pr[q_{{\sf rep}} = t] \cdot \Pr[{\sf type}_{\bar{V}}] \cdot {T \choose q_{{\sf rep}}} \cdot {k \choose q_{{\sf rep}}} \cdot \\ 
    & \ \ \ \ \ \  \cdot \Pr[\cA \text{ wins } | \ v_1, ..., v_k, {\sf type}] \\
    &= k \cdot \sum_{v_1, ..., v_k}  \sum_{{\sf type}_{\bar{V}}} \sum_{t = 0}^k p^{k - t} \cdot (1 - p)^t \cdot \Pr[{\sf type}_{\bar{V}}] \cdot {T \choose t} \cdot {k \choose t} \cdot \\
   & \ \ \ \ \ \  \cdot \Pr[\cA \text{ wins } | \ v_1, ..., v_k, {\sf type}] \\
    &= k \cdot \sum_{t = 0}^k p^{k - t} \cdot (1 - p)^t \cdot  {T \choose t} \cdot {k \choose t} \cdot \sum_{v_1, ..., v_k}  \sum_{{\sf type}_{\bar{V}}} \Pr[{\sf type}_{\bar{V}}]  \cdot \\
    & \ \ \ \ \ \  \cdot \Pr[\cA \text{ wins } | \ v_1, ..., v_k, {\sf type}]
\end{align*}
Similarly, we can now express the success probability of the simulator $\cB$:
\begin{align*}
    \Pr[\cB \text{ wins}]
    &\geq \frac{1}{{T \choose k}} \sum_{v_1, ..., v_k} \Pr[\cB \text{ wins} \ | \   v_1, ..., v_k] \\
     &= \frac{1}{{T \choose k}} \sum_{v_1, ..., v_k} \sum_{{\sf type}} 
    \Pr[{\sf type}] \Pr[\cB \text{ wins} \ | \  v_1, ..., v_k, {\sf type}] \\
     &= \frac{1}{{T \choose k}} \sum_{v_1, ..., v_k} \sum_{{\sf type}_{\bar{V}}} \sum_{t = 0}^k p^{k - t} \cdot (1 - p)^t \cdot \Pr[{\sf type}_{\bar{V}}] \cdot \Pr[\cB \text{ wins } | \ v_1, ..., v_k, {\sf type}] \\
     &= \frac{1}{{T \choose k}} \sum_{t = 0}^k p^{k - t} \cdot (1 - p)^t \sum_{v_1, ..., v_k} \sum_{{\sf type}_{\bar{V}}}  \Pr[{\sf type}_{\bar{V}}] \cdot \Pr[\cB \text{ wins } | \ v_1, ..., v_k, {\sf type}]
\end{align*}

Now, as explained, using the fact that: \\
$\Pr[\cB \text{ wins } | \ v_1, ..., v_k, {\sf type}] =  \Pr[\cA \text{ wins } | \ v_1, ..., v_k, {\sf type}]$, we also have that $X := \sum_{v_1, ..., v_k} \sum_{{\sf type}_{\bar{V}}}  \Pr[{\sf type}_{\bar{V}}] \cdot \Pr[\cA \text{ wins } | \ v_1, ..., v_k, {\sf type}] = \sum_{v_1, ..., v_k} \sum_{{\sf type}_{\bar{V}}}  \Pr[{\sf type}_{\bar{V}}] \cdot \Pr[\cB \text{ wins } | \ v_1, ..., v_k, {\sf type}]$. Therefore, we can connect the probabilities of winning for the algorithm and simulator as follows:
\begin{align*}
    \Pr[\cA \text{ wins}]   
    & \leq k \cdot \sum_{t = 0}^k p^{k - t} \cdot (1 - p)^t \cdot  {T \choose t} \cdot {k \choose t} \cdot X \\
    & \leq k \cdot \sum_{t=0}^k (1-p)^{t} \cdot {T \choose t} \cdot {k \choose t} \cdot X \\
    & \leq (((1-p) T k)^{k} + k) \cdot X
\end{align*}
\begin{align*}
    \Pr[\cB \text{ wins}]   
    &\geq \frac{1}{{T \choose k}} \sum_{t = 0}^k p^{k - t} \cdot (1 - p)^t \cdot X 
    =
    \begin{cases}
        & \frac{1}{{T \choose k}} \cdot \frac{p^{k + 1} - (1 - p)^{k+1}}{2p - 1} \cdot X \text{ , if } p \neq \frac{1}{2} \\
        & \frac{1}{{T \choose k}} \cdot \frac{k + 1}{2^k} \cdot X \text{ , if } p = \frac{1}{2}
    \end{cases}
\end{align*}

Therefore, by combining both inequalities, we have:
\begin{align*}
    \Pr[\cA \text{ wins}] &\leq  
    \Pr[\cB \text{ wins}] \cdot 
        O\left( {T \choose k} \left( ((1-p) T k)^k + k \right) \right) 
\end{align*}

\end{proof}

\section{Applications} \label{sec:applications}

{In this section, we 
show a series of applications of 
our hybrid coherent measure-and-reprogram theorem (\Cref{thm:hybrid_coherent_measure_and_reprogram}) in query complexity and cryptography. 

\subsection{Query Complexity}

We will begin by first introducing the family of (security) games for which we will establish their quantum {and hybrid} query complexity, namely, the hardness of a quantum {or a hybrid} adversary to win such games.

\begin{definition}[Multi-Output $k$-Search Game (Single-Instance)] 
    \label{def:multi_instance_game}
Let the random oracle $H : [M] \rightarrow [N]$,  a distribution over challenges $\pi_H$ and a winning relation $R_{H, ch}$ defined over $Y^k$. 
Then we define the \multigame $\gamemath$ as follows:
\begin{enumerate}
    \item Challenger samples randomness $ch$ and sends it to a quantum algorithm $\cA$ having (quantum) oracle access to $H$;
    \item Adversary $\cA$ outputs 
    $\vec{x} := (x_1, ..., x_k), z$;
    \item Challenger queries $\vec{x}$ to the random oracle, resulting in $\vec{y} := (y_1 = H(x_1), ..., y_k = H(x_k))$ and
    checks if they satisfy the winning relation: \\
    $b := ((x_1, \ldots, x_k, y_1, \ldots, y_k, z) \in R_{H, ch})$;
    \item If $b = 1$, $\cA$ wins the $\gamemath$ game.
\end{enumerate}
We will denote by $\epsilon_{\gamemath}(q)$ the maximum probability over all $q$-quantum algorithms $\cA$ of winning the \multigame $\gamemath$.
\end{definition}
{\begin{remark} We remark that in Definition~\ref{def:multi_instance_game}
the winning relation is allowed to also depend on the challenge and the RO in order to provide a more general framework, while 
in fact, not all applications need this dependency.
\end{remark}}

Our main result is a quantum lifting theorem in the average case, relating the success probability of an arbitrary hybrid algorithm to win a \multigame with the probability of success of a hybrid algorithm equipped with exactly $k$ queries.

\begin{lemma}[Hybrid Lifting for Multi-Output $k$-Search Games]
\label{lem:hybrid_lifting}
Let $\gamemath$ be a 
multi-output $k$-search game.
Let $\cA$ be a hybrid adversary equipped with $q$ quantum and $c$ classical queries in $\gamemath$ (against the $k$-classical query challenger $\cC$).
Then there exists a $k$ (quantum + classical) query hybrid adversary $\cB$ against the game such that:
\begin{equation*}
\Pr[\cB^{H} \text{ wins } \gamemath] \geq \frac{1}{\left(O\left( \frac{q^2}{k^2} + \frac{c}{k} \right)\right)^{k}} \Pr[\cA^{H} \text{ wins } \gamemath].   
\end{equation*}
\end{lemma}

Let $L_{\cC}$ represent the set of classical queries that the challenger performs during a \multigame $\gamemath$ (Def.~\ref{def:multi_instance_game}). For a quantum query adversary $\cB$ against $\gamemath$, we denote by $L_{\cB}$ the result of measuring its input and output query registers.
Now, for the query complexity applications we will need the following stronger lifting theorem, which intuitively additionally guarantees the existence of an algorithm against $\gamemath$ such that at the end of the game, measuring its input and output registers gives us exactly the set of queries of the challenger.

\begin{lemma}[Hybrid Lifting for Search Game with Uniform Images]
\label{lem:hybrid_lifting_improved}
Let $\gamemath$ be a multi-output $k$-search game.
Let $\cA$ be a hybrid adversary equipped with $q$ quantum and $c$ classical queries in $\gamemath$ (against $k$-classical query challenger $\cC$).
Then there exists a $k$ (quantum + classical) query hybrid adversary $\cB$ such that $L_{\cB}$ is uniform, satisfying:
\begin{align*}
    \Pr[\cB^{H} \text{ wins } \gamemath \text{ and } L_{\cC} = L_{\cB}] \geq \frac{1}{\left(O\left( \frac{q^2}{k^2} + \frac{c}{k} \right)\right)^{k}} \Pr[\cA^{H} \text{ wins } \gamemath]. 
\end{align*}
\end{lemma}

\subsection{Hybrid Lifting and Direct Product Theorems for Image Relations}

Our first quantum lifting result (in \Cref{lem:hybrid_lifting}) gives a first bound on the quantum hardness of solving any \multigame $\gamemath$ by relating it to the probability of $\gamemath$ being solved by a quantum algorithm with a small number of quantum queries.
In this section, we can derive  stronger hybrid lifting theorems for the class of relations that only depend on images.

\begin{lemma}[Hybrid Lifting Theorem for Image Relations]
\label{thm:hybrid_direct_product_strong}
    For any hybrid algorithm $\cA$ equipped with $q$ quantum
    and $c$ classical queries, $\cA$'s success probability to solve the \multigame specified by a winning relation $R$, is bounded by:
    \begin{equation*}
        \begin{split}
        \Pr[\cA^{H} & \text{ wins \multigame}] \leq 
         \left(O\left( \frac{q^2}{k^2} + \frac{c}{k} \right)\right)^{k} \cdot p(R).
        \end{split}
    \end{equation*}
\end{lemma}

Next, we show a Direct Product Theorem for image relations. 
\begin{definition}[Direct Product]
    Let $\gamemath$ be a multi-output $k$-search game specified by the winning relation $R$, with respect to a random oracle $[M] \to [N]$. Define the following Direct Product $\gamemath^{\otimes g}$: 
    \begin{itemize}
        \item Let $H$ be a random oracle $[g] \times [M] \to [N]$, and $H_i$ denotes $H(i, \cdot)$;  
        \item Challenger samples $ch_i$ as in $\gamemath$ for $i \in \{1, \ldots, g\}$. 
        \item Adversary $\As$ gets oracle access to $H$ and outputs $\vec{x}_1, \ldots, \vec{x}_g$, $z_1, \ldots, z_g$ such that each input in $\vec{x}_i$ starts with $i$.  
        \item Challenger computes $b_i := (\vec{x}_i, H(\vec{x}_i), z_i) \in R_{H_i, ch_i}$;
        \item If all $b_i$ equal to $1$, $\As$ wins the $\gamemath^{\otimes g}$ game. 
    \end{itemize}
\end{definition}

{
\begin{lemma}[Hybrid Direct Product Theorem for Image Relations]
\label{corr:hybrid_dpt}
    For any hybrid algorithm $\As$ equipped with $g q$ quantum and $g c$ classical queries, $\As$'s success probability to solve the Direct Product $\gamemath^{\otimes g}$ with the underlying $\gamemath$ specified by the winning relation $R$, is bounded by:
\begin{align*}
    \Pr[\As^{H} & \text{ wins } G^{\otimes g}] \leq \left( \left(O\left( \frac{q^2}{k^2} + \frac{c}{k} \right)\right)^{k} p(R) \right)^g.
\end{align*}
\end{lemma}
}
    
We defer additional applications, including the query complexity and cryptographic implications of our hybrid lifting theorems and Direct Product Theorem, to \Cref{sec:other_app} in the Supplementary Material.

\section*{{Acknowledgements}}
J.G. was partially supported by NSF SaTC grants no. 2001082 and 2055694.
F.S. was partially supported by NSF grant no. 1942706 (CAREER). J.G. and F.S. were also partially support by Sony by means of the Sony Research Award Program.
A.C. acknowledges support from the National Science Foundation grant CCF-1813814, from the AFOSR under Award Number FA9550-20-1-0108 and the support of the Quantum Advantage Pathfinder (QAP) project, with grant reference EP/X026167/1 and the UK Engineering and Physical Sciences Research Council.

\newpage

\printbibliography

\appendix

\vspace{.2in}
\begin{center}
    {\Large {\bf Supplementary Material}}
\end{center}

\section{Other Applications}\label{sec:other_app}
\subsection{Application 1: Non-uniform Security}

\begin{definition}[Advice Algorithms]
We define an advice (non-uniform) algorithm $\cA = (\cA_1, \cA_2)$ equipped with $q + c$ queries and advice of length $S$ as follows:
\begin{enumerate}
    \item $\cA_1^{H} \rightarrow \ket{adv}$: an unbounded algorithm $\cA_1$ outputs the advice $\ket{adv}$ consisting of $S$ qubits;
    \item $\cA_2^{H}(\ket{adv}, ch) \rightarrow x$: $q$-quantum algorithm $\cA_2$ takes as input the quantum advice $\ket{adv}$ and a challenge $ch$, outputs answer $x$;
\end{enumerate}
    Similarly, we define $\epsilon_{\gamemath}^C(q, c, S)$ as the maximum winning probability over any advice hybrid adversary $\cA$ equipped with $q$ quantum queries, $c$ classical queries and $S$ classical bits of advice against the classically-verifiable search game $\gamemath$. For purely quantum adversaries, we will simply denote the maximum probability by  $\epsilon_{\gamemath}^C(q, S)$.
\end{definition}

We also consider multi-instance games, similar to Direct Product, except all the instances share the same oracle. 
\begin{definition}[Multi-Instance Game]
    Let $\gamemath$ be a multi-output $k$-search game specified by the winning relation $R$, with respect to a random oracle $[M] \to [N]$. Define the following Direct Product $\gamemath_{\sf MIS}^{\otimes g}$: 
    \begin{itemize}
        \item Let $H$ be a random oracle $[M] \to [N]$; 
        \item Challenger samples $ch_i$ as in $\gamemath$ for $i \in \{1, \ldots, g\}$;
        \item Adversary $\As$ gets oracle access to $H$ and outputs $\vec{x}_1, \ldots, \vec{x}_g$, $z_1, \ldots, z_g$; 
        \item Challenger computes $b_i := (\vec{x}_i, H(\vec{x}_i), z_i) \in R_{H, ch_i}$;
        \item If all $b_i$ equal to $1$, $\As$ wins the $\gamemath_{\sf MIS}^{\otimes g}$ game. 
    \end{itemize}
    From above, we can define $R^{\otimes g}_{\sf MIS}$ as the winning relation for $\gamemath_{\sf MIS}^{\otimes g}$.
\end{definition}

\begin{lemma}[Security against Advice Hybrid Adversaries]\label{lemma:advice_hybrid}
Let $\gamemath$ be any multi-output $k$-search game specified by the winning relation $R$. Let $\gamemath^{\otimes g}_{\sf MIS}$ be the multi-instance game and $R^{\otimes g}_{\sf MIS}$ be the relation. 
Any non-uniform hybrid algorithm $\cA$ equipped with $q$ quantum queries, $c$ classical queries and $S$ classical bits of advice can win the game $\gamemath$ with probability at most:
\begin{align*}   
    \epsilon_{\gamemath}^C(q, c, S) & 
    \leq
    \left(O\left( \frac{S^2 q^2}{k^2} + \frac{Sc}{k} \right)\right)^{\frac{k}{S}}
    \cdot p(R_{\sf MIS}^{\otimes S}) \, .
\end{align*} 
\end{lemma}

\noindent To show \Cref{lemma:advice_hybrid}, we will rely on the following result:
\begin{lemma}[Multi-Output Implies Non-Uniform Classical Advice (\cite{chung2020tight})]
\label{lemma:salting_classical}
    Let $\gamemath$ be a search game (as defined in Def.~\ref{def:multi_instance_game}).
    If the maximum winning probability for any quantum algorithm equipped with $q$ quantum queries against $\gamemath^{\otimes g}_{\sf MIS}$ is $\epsilon_{\gamemath_{\sf MIS}^{\otimes g}}(q)$, then the maximum probability of any non-uniform adversary equipped with $q$ quantum queries and $S$-length classical advice against the original game $\gamemath$ is at most:
    \begin{equation*}
        \epsilon_{\gamemath}^C(q, S) \leq 4 \cdot \left[\epsilon_{\gamemath_{\sf MIS}^{\otimes {S}}}(Sq)\right]^{\frac{1}{S}}
    \end{equation*}
\end{lemma}

Then, \Cref{lemma:advice_hybrid} follows directly by combining our (strong) hybrid lifting theorem (in \Cref{thm:hybrid_direct_product_strong}) with the advice result (Lemma~\ref{lemma:salting_classical}).

\subsection{Application 2: Salting Against Non-uniform Adversaries}

\begin{definition} [Salted Game] Let $\gamemath$ be a search game (as defined in Def.~\ref{def:multi_instance_game}) specified by a random oracle $H : [M] \rightarrow [N]$, a distribution over challenges $\pi_H$ and a winning relation $R_{H, ch}$ defined over $Y$. Then we define the salted version of $\gamemath$ as the game $\gamemath_s$ with salted space $[K]$ defined as follows:
\begin{enumerate}
    \item The random oracle function is defined as: $G = (H_1, ..., H_K)$ for $K$ random functions $H_i : [M] \rightarrow [N]$;
    \item For any such $G$, the challenge $ch := (i, ch_i)$ is produced by first sampling uniformly at random $i \in [K]$ and then sampling $ch_i$ according to $\pi_{H_i}$;
    \item The winning relation is defined as $R_{G, ch} := R_{H_i, ch_i}$;
\end{enumerate}
We will denote by 
$\epsilon_{\gamemath_s}(q, c)$
the maximum probability over all
$q$-quantum and $c$-classical hybrid algorithms
$\cA$ of winning the salted game $\gamemath_s$. 
\end{definition}

\begin{lemma}[Hybrid Security of Salted Game against Classical Advice]
\label{lemma:hybrid_salted_multi_advice}
Let $\gamemath$ be a multi-output $k$-search game (as defined in Def.~\ref{def:multi_instance_game}), specified by a relation $R$. Let $\gamemath_s$ be the salted game, with salt space $[K]$. Then we have,
\begin{align*}
    \epsilon^C_{\gamemath_s}(q, c, S) \leq 4 \cdot \frac{S}{K} + \left(O\left( \frac{q^2}{k^2} + \frac{c}{k} \right)\right)^{k} \cdot p(R). 
\end{align*}
\end{lemma}

\begin{proof}
By \Cref{lemma:salting_classical}, the non-uniform security is related to the multi-instance game $\gamemath_{s, \sf MIS}^{\otimes g}$, with salt space $[K]$. The security of the multi-instance game is closely related to the Direct Product, for salted games, as shown in \cite{dong2024salting} (in the proof of Theorem 4.1). More precisely, 
\begin{align*}
    \epsilon_{\gamemath_{s, \sf MIS}^{\otimes g}}(g q)^{1/g} \leq \epsilon_{\gamemath_{s}^{\otimes g}}(g q)^{1/g} + \frac{g}{K}. 
\end{align*}
Intuitively, the only difference between the multi-instance game and the Direct Product is that, the same salt can be sampled with duplication. The extra factor $\frac{g}{K}$ captures the fact that the salt can be duplicated. 
Combining with \Cref{corr:hybrid_dpt}, we have:
\begin{align*}
    \epsilon_{\gamemath_{s}}^C(q, S) & \leq 4 \left( \epsilon_{\gamemath_{s, \sf MIS}^{\otimes S}} (S q) \right)^{1/S} \\
     & \leq 4 \left(\epsilon_{\gamemath_{s}^{\otimes S}}(S q)^{1/S} + \frac{S}{K} \right) \\
     & \leq 4 \cdot \frac{S}{K} + \left(O\left( \frac{q^2}{k^2} + \frac{c}{k} \right)\right)^{k}. 
\end{align*}

\end{proof}

\subsection{Application 3: Multi-Image Inversion}

\hfill \break

Our first result establishes the first hybrid hardness of multi-image inversion. Additionally, we provide a tight bound, as proven earlier in \Cref{subsec:hybrid_optimal_multi_search}.

\begin{lemma}[Hybrid Hardness and Optimality of Multi-Image Inversion]
For any hybrid algorithm $\cA$ equipped with $q$ quantum and $c$ classical queries the success probability of $\cA$ to solve the multi-image inversion problem is upper bounded by:
    \begin{align*}
         & \Pr_H[\cA^{H}(\vec{y}) \rightarrow \vec{x} = (x_1, ..., x_k) \ : \ H(x_i) = y_i \ \forall i \in [k]] \nonumber \\
             \leq &  
              \left(O\left( \frac{q^2}{k^2} + \frac{c}{k} \right)\right)^{k} \cdot \frac{k!}{N^k}
    \end{align*}
Moreover, there exists a hybrid algorithm whose success probability matches up to constant factors this bound (as proven in \Cref{lemma:hybrid_alg_multi_image} and \Cref{lemma:optimal_hybrid_multi_image}).
\end{lemma}
\begin{proof}
We will show this using our strong hybrid lifting theorem for image relations (\Cref{thm:hybrid_direct_product_strong}).
Define $R$ as the relation over $[N]^k$, with $H: [M] \rightarrow [N]$ such that: $R = \{y_1, ..., y_k\}$. 
Then for each permutation $\pi$, we have $\Pr[(y_{\pi(1)}, ..., y_{\pi(k)}) \in R \ | \ (y_1, ..., y_k) \leftarrow [N]^k] = \frac{1}{N^k}$. Using that the number of permutations $\pi$ is $k!$ leads to:
\begin{align*}
    \Pr_H[\cA^{H}(\vec{y}) \rightarrow \vec{x} = (x_1, ..., x_k) \ : \ H(x_i) = y_i \ \forall i \in [k]] \leq 
    \left(O\left( \frac{q^2}{k^2} + \frac{c}{k} \right)\right)^{k}
    \cdot \frac{k!}{N^k}
\end{align*}

\end{proof}

\subsection{Application 4: Multi-Collision Finding, Multi-Search and 3SUM}

\hfill \break

Next, we can determine the hybrid hardness of the multi-collision problem, namely finding $k$ different inputs that map to the same output of the random oracle.

\begin{lemma}[Hybrid Hardness of Multi-Collision Finding]
\label{lemma:hybrid_hardness_multi_collision}
    For any hybrid algorithm $\cA$ equipped with $q$ quantum and $c$ classical queries, the success probability of $\cA$ to solve the $k$-multi-collision problem is at most:
        \begin{equation*}
        \Pr_H[\cA^{H}() \rightarrow \vec{x} = (x_1, ..., x_k) \ : \ H(x_1) = ... = H(x_k)] \leq 
         \left(O\left( \frac{q^2}{k^2} + \frac{c}{k} \right)\right)^{k} \cdot \frac{1}{N^{k - 1}}
    \end{equation*}
\end{lemma}
\begin{proof}
We will show this using our strong quantum lifting theorem for image relations (\Cref{thm:hybrid_direct_product_strong}).
Define $R := \{y, ..., y\}_y$ the relation over $[N]^k$, where $H: [M] \rightarrow [N]$. Then for each permutation $\pi$, we have $\Pr[(y_{\pi(1)}, ..., y_{\pi(k)}) \in R \ | \ (y_1, ..., y_k) \leftarrow [N]^k] = \frac{1}{N^{k-1}}$. As $R$ is permutation invariant, this implies that:
    \begin{equation*}
        \Pr_H[\cA^{H}() \rightarrow \vec{x} = (x_1, ..., x_k) \ : \ H(x_1) = ... = H(x_k)] \leq 
         \left(O\left( \frac{q^2}{k^2} + \frac{c}{k} \right)\right)^{k}
         \cdot \frac{1}{N^{k - 1}}
    \end{equation*}
\end{proof}

Next, we consider another search application, namely the task of determining $k$ different inputs that all map to $0$ under the random oracle. One of the main motivations behind this problem is its relation to the notion of proof-of-work in the blockchain context~\cite{garay2015bitcoin}.

\begin{lemma}[Hybrid Hardness of Multi-Search]\label{lemma:hybrid_multi_search}
  For any hybrid algorithm $\cA$ equipped with $q$ quantum and $c$ classical queries whose task is to find different $k$ preimages of $0$ of a random oracle $H$, the success probability of $\cA$ is upper bounded by:
    \begin{align*}
         \Pr_H[\cA^{H}() \rightarrow \vec{x} = (x_1, ..., x_k) \ : \ H(x_i) = 0 \ \forall i \in [k]]  
             \leq 
             \left(O\left( \frac{q^2}{k^2} + \frac{c}{k} \right)\right)^{k} \cdot \frac{1}{N^k}
    \end{align*}
    
\end{lemma}
\begin{proof}
    We will show this using our strong quantum lifting theorem for image relations (\Cref{thm:hybrid_direct_product_strong}).
Define $R := \{0, ..., 0\}$ the relation over $[N]^k$, where $H: [M] \rightarrow [N]$. Then for each permutation $\pi$, we have $\Pr[(y_{\pi(1)}, ..., y_{\pi(k)}) \in R \ | \ (y_1, ..., y_k) \leftarrow [N]^k] = \frac{1}{N^{k}}$. As $R$ is permutation invariant, this implies that:
  \begin{align*}
    \Pr_H[\cA^{H}() \rightarrow \vec{x} = (x_1, ..., x_k) \ : \ H(x_i) = 0 \ \forall i \in [k]] 
             \leq 
             \left(O\left( \frac{q^2}{k^2} + \frac{c}{k} \right)\right)^{k}
             \frac{1}{N^k}
    \end{align*}

\end{proof}

Note that this bound is asymptotically tight, as an algorithm with $q$ queries can use $q/k$ queries to find each pre-image (Grover's algorithm), resulting in a probability of $\Theta\left( \left\{ (\frac{q}{k})^2/N \right\}^k\right)$.

Finally, we consider the 3SUM computational problem. In this task an algorithm needs to find 3 inputs
that are mapped by the random oracle to 3 images
that must sum up to $0$.
As before, to establish the quantum and hybrid complexity of the problem one only needs to compute the classical quantity $p(R)$.

\begin{lemma}[Hybrid Hardness of 3SUM]\label{lemma:hybrid_3sum}
       For any hybrid algorithm $\cA$ equipped with $q$ quantum and $c$ classical queries, the success probability of $\cA$ to solve the 3SUM problem is upper bounded by:
         \begin{align*}
    \Pr_H &[\cA^{H}() \rightarrow \vec{x} = (x_1, x_2, x_3) \ : \ H(x_1) + H(x_2) + H(x_3) = 0] \\
             &\leq O \left(\left( q^2 + c \right)^{3} \cdot \frac{3N^2 + 3N + 1}{(2N + 1)^3} \right)
    \end{align*}
\end{lemma}

\begin{proof}
Denote the set $Y :=  \{-N, -N+1, ..., 0, ..., N\}$. Then for a random oracle $H: [M] \rightarrow Y$, the relation $R$ is defined as follows: $R = \{(y_1, y_2, y_3) \ | \ y_1 + y_2 + y_3 = 0 \}$.
For each permutation $\pi$, we have:
\begin{align*}
    \Pr&[(y_{\pi(1)}, y_{\pi(2)}, y_{\pi(3)}) \in R \ | \ (y_1, y_2, y_3) \leftarrow Y^3] \\
    &= \Pr[y_3 = -(y_1 + y_2) \text{ and } -N \leq y_1 + y_2 \leq N] \\
    &= \frac{1}{2N + 1} \cdot \frac{3N^2 + 3N + 1}{(2N +1)^2}.    
\end{align*}
As $R$ is permutation invariant, this implies that for any hybrid algorithm $\cA$ equipped with $q$ quantum and $c$ classical queries:
    \begin{align*}
    \Pr[\cA \text{ solves 3SUM}] 
             \leq O \left(\left( q^2 + c \right)^{3} \cdot \frac{3N^2 + 3N + 1}{(2N + 1)^3} \right)
    \end{align*}
\end{proof}

\section{Deferred Proofs} \label{app:proofs}

\subsection{Proof of \Cref{lem:hybrid_lifting}}

\begin{proof}
We will show that our Coherent Reprogramming result in Theorem~\ref{thm:hybrid_coherent_measure_and_reprogram} implies the lifting theorem. We will now show how to instantiate the coherent reprogramming theorem. Let $\xout$ be uniformly sampled from $X^k$. Let $H', G' : X \rightarrow Y$ be two uniform random oracles. Then, it is clear that, as $\yout = G'(\xout)$ is also uniform over $Y^k$, the reprogrammed function ${H'}_{\xout, \yout} : X \rightarrow Y$ is a uniform random function; this is due to the fact 
that, as stated in Theorem~\ref{thm:hybrid_coherent_measure_and_reprogram} (invoked here), the tuple $\xout$ has distinct values for its element.
We will instantiate the random oracle in the game $\gamemath$ as the function ${H'}_{\xout, \yout}$. Assume that in the game $\gamemath$ after receiving the challenge and after performing its $q$ quantum 
 and $c$ classical queries to ${H'}_{\xout, \yout}$, the adversary $\cA$ returns to the Challenger the outcome $\vec{x}$. Then, the Challenger queries $\vec{x}$ to ${H'}_{\xout, \yout}$ resulting in $\vec{y}$ and checks if $\vec{y}$ satisfies the winning relation $R_{{H'}_{\xout, \yout}, ch}$. Define $V^{{H'}_{\xout, \yout}}$ as the predicate that outputs $1$ if $\vec{y} \in R_{{H'}_{\xout, \yout}, ch}$ and $0$ else. In this way, we observe that the probability that $\cA$ wins the game $\gamemath$ is exactly the RHS of Theorem~\ref{thm:hybrid_coherent_measure_and_reprogram}. As a result, by Theorem~\ref{thm:hybrid_coherent_measure_and_reprogram}, there must exist an efficient quantum simulator $\Sim^{H', G', \cA}$ performing $k$ quantum and classical queries that also wins the game $\gamemath$. Hence, it suffices to instantiate $\cB$ as the simulator $\Sim$.
\end{proof}

\subsection{Proof of \Cref{lem:hybrid_lifting_improved}}

\begin{proof}
     The simulator algorithm $\cB$ will follow the outline of the algorithm in the proof of \Cref{lem:hybrid_lifting}, with the only difference that $\cB$ will perform an additional step at the end. 
Namely, after interaction with Challenger $\cC$, compute list of queries of $\cC$ as $L_{\cC}$. If any query in $L_{\cC}$ has not yet been queried by $\cB$, $\cB$ will query them to oracle $H$. 
\end{proof}

\subsection{Proof of \Cref{thm:hybrid_direct_product_strong}}

\begin{proof}
    Let $\gamemath$ be a \multigame and assume a hybrid adversary $\cA$ equipped with $q$ quantum and $c$ classical queries that sends to the Challenger the answer $\vec{x} = (x_1, ..., x_k)$. Challenger $\cC$ will accept if and only if $\vec{y} := (H(x_1), ..., H(x_k)) \in R_{H, ch}$ and if $x_i, x_j$ are pairwise distinct.
    By \Cref{lem:hybrid_lifting_improved} we know there exists a hybrid algorithm $\cB$ making $k$ quantum and classical queries to $H$ winning the game such that $L_{\cB} = L_{\cC}$ with success probability at least the success probability of $\cA$ multiplied by a factor of
    $ \left(O\left( \frac{q^2}{k^2} + \frac{c}{k} \right)\right)^{k}$.
    The condition $L_{\cB} = L_{\cC}$ implies that $\cC$ will verify as the images of $\cB$'s answer exactly a permutation of the recorded information in $L_{\cB}$. Therefore, due to the property of \Cref{lem:hybrid_lifting_improved} that $L_{\cB}$ will be uniformly over $Y^k$, $\cB$'s winning probability will be lower bounded by the probability that there exists a permutation such that for uniformly sampled images from $Y^k$, the permuted images will belong to our target relation:
    \begin{align*}
    \Pr[\cA^{H} & \text{ wins \multigame}] \leq 
                 \left(O\left( \frac{q^2}{k^2} + \frac{c}{k} \right)\right)^{k} p(R) \, .
    \end{align*}
\end{proof}

\subsection{Proof of \Cref{corr:hybrid_dpt}}

\begin{proof}
Let $\gamemath$ be a \multigame and assume a hybrid adversary $\cA$ equipped with $g q$ quantum and $g c$ classical queries that sends to the Challenger the answer $\vec{x}_1, \ldots, \vec{x}_g, z_1, \ldots, z_g$.
By \Cref{lem:hybrid_lifting_improved} we know there exists a quantum algorithm $\cB$ making $gk$ quantum and classical queries to $H$ winning the game such that $L_{\cB} = L_{\cC}$ with success probability at least the success probability of $\cA$ multiplied by a factor of
    $\left(O\left( \frac{(gq)^2}{k^2} + \frac{gc}{k} \right)\right)^{k}$.
    The condition $L_{\cB} = L_{\cC}$ implies that $\cC$ will verify as the images of $\cB$'s answer exactly a permutation of the recorded information in $L_{\cB}$. Moreover, for every image $y$, its associated input $x$ only belongs to one of the oracles $H(i, \cdot)$; thus, it can only contribute to one of the relation checks $R_{H_i, ch_i}$. Thus, the permutation of the recorded information can only permute images with respect to the same oracle $H_i$.

    Therefore, due to the property of \Cref{lem:hybrid_lifting_improved} that $L_{\cB}$ will be uniformly over $Y^{gk}$, $\cB$'s winning probability will be lower bounded by the probability that there exists a permutation such that for uniformly sampled images from $Y^{gk}$, the permuted images will belong to our target relation:
    \begin{align*}               
    & \Pr[\cA^{H} \text{ wins } \gamemath^{\otimes g}] \\ & \leq \left(O\left( \frac{(gq)^2}{k^2} + \frac{gc}{k} \right)\right)^{k}\Pr[\exists  \pi_1,\ldots,\pi_g \in \sym_k \ | \ (y_{i, {\pi_i(1)}}, ..., y_{i, {\pi_i(k)}}) \in R  : (y_{i, 1}, ..., y_{i, k}) \xleftarrow{\$} Y^k] \\
                & \leq  \left(  \left(O\left( \frac{q^2}{k^2} + \frac{c}{k} \right)\right)^{k} \Pr[\exists \text{  } \pi \in \sym_k \ | \ (y_{\pi(1)}, y_{\pi(2)}, ..., y_{\pi(k)}) \in R  : (y_1, ..., y_k) \xleftarrow{\$} Y^k] \right)^g \\
                & \leq  \left(  \left(O\left( \frac{q^2}{k^2} + \frac{c}{k} \right)\right)^{k} p(R)\right)^g \, .
    \end{align*}    
\end{proof}

\end{document}